%% file: main.tex
\documentclass{article}[11pt]
\usepackage{fullpage}

\input{macros}

\title{A Multi-Dimensional Online Contention Resolution Scheme for Revenue Maximization\footnote{This work was funded in part by NSF award CCF-2217069.}}
\author{
Shuchi Chawla \\ {\tt shuchi@cs.utexas.edu} \and 
Dimitris Christou  \\ {\tt christou@cs.utexas.edu} \and
Trung Dang  \\ {\tt dddtrung@cs.utexas.edu}
\and
Zhiyi Huang \\ {\tt zhiyih@cs.utexas.edu}
\and
Gregory Kehne  \\ {\tt gkehne@utexas.edu} \and
Rojin Rezvan  \\ {\tt rojin@cs.utexas.edu}\footnote{All co-authors are affiliated with the University of Texas at Austin.}
}

\date{}

\begin{document}

\maketitle

\thispagestyle{empty}
\addtocounter{page}{-1}

\begin{abstract}
    We study multi-buyer multi-item sequential item pricing mechanisms for revenue maximization with the goal of approximating a natural fractional relaxation -- the ex ante optimal revenue. We assume that buyers' values are subadditive but make no assumptions on the value distributions. While the optimal revenue, and therefore also the ex ante benchmark, is inapproximable by any simple mechanism in this context, previous work has shown that a weaker benchmark that optimizes over so-called ``buy-many" mechanisms can be approximable. Approximations are known, in particular, for settings with either a single buyer or many unit-demand buyers. We extend these results to the much broader setting of many subadditive buyers. We show that the ex ante buy-many revenue can be approximated via sequential item pricings to within an $O(\log^2 m)$ factor, where $m$ is the number of items. We also show that a logarithmic dependence on $m$ is necessary.

    Our approximation is achieved through the construction of a new multi-dimensional Online Contention Resolution Scheme (OCRS), that provides an online rounding of the optimal ex ante solution. \citet{chawla23buy} previously constructed an OCRS for revenue for unit-demand buyers, but their construction relied heavily on the ``almost single dimensional" nature of unit-demand values. Prior to that work, OCRSes have only been studied in the context of social welfare maximization for single-parameter buyers. For the welfare objective, constant-factor approximations have been demonstrated for a wide range of combinatorial constraints on item allocations and classes of buyer valuation functions. Our work opens up the possibility of a similar success story for revenue maximization.
\end{abstract}

\newpage
\input{sec_intro}
\input{sec_prelims}
\input{sec_subadd}
\input{sec_rrs_argument}
\input{sec_ocrs_argument}
\input{sec_removing_pmax_pmin}
\input{sec_computational_aspects}

\input{sec_xos_lb}
\input{sec_related}
\input{sec_conclusion}

\newpage
\bibliographystyle{plainnat}
\bibliography{references} 

\newpage
\appendix
\input{sec_prelims_MD}
\input{sec_monotone}

\input{sec_gs_proof}
\input{sec_rrs_lowbound}
\input{sec_xos_lb_proofs}

\end{document}

%% file: macros.tex
\usepackage{xcolor}
\usepackage{dirtytalk}
\usepackage{amsfonts}
\usepackage{amsthm}
\usepackage{amsmath}
\usepackage{amssymb}
\usepackage{mathtools}
\usepackage{bbm}
\usepackage[hidelinks]{hyperref}
\usepackage{cleveref}
\usepackage{enumerate}
\usepackage{comment}
\usepackage{algorithm}
\usepackage[noend]{algpseudocode}

\usepackage[]{natbib}
\usepackage{thmtools}
\usepackage{thm-restate}
\renewcommand{\cite}[1]{\citep{#1}}
\newcommand{\textcite}[1]{\citet{#1}}

\newtheorem{theorem}{Theorem}
\newtheorem{inftheorem}{Informal Theorem}
\newtheorem{fact}{Fact}
\newtheorem{claim}[theorem]{Claim}
\newtheorem{lemma}[theorem]{Lemma}

\newtheorem{corollary}[theorem]{Corollary}
\newtheorem{proposition}[theorem]{Proposition}

\theoremstyle{definition}
\newtheorem{definition}{Definition}
\newtheorem*{infdef}{Definition}

\newcommand{\rev}{\textsc{Rev}}
\newcommand{\sw}{\textsc{SW}}
\newcommand{\earev}{\textsc{EARev}}
\newcommand{\obj}{\textsc{Obj}}
\newcommand{\objs}{\textsc{S-Obj}}
\newcommand{\objl}{\textsc{L-Obj}}
\newcommand{\objm}{\textsc{M-Obj}}

\newcommand{\D}{\mathcal{D}}

\newcommand{\C}{\mathcal{C}}
\newcommand{\R}{\mathbb{R}}
\newcommand{\F}{\mathbb{D}}
\newcommand{\E}{\mathcal{E}}
\newcommand{\M}{\mathcal{M}}
\newcommand{\Rplus}{\mathbb{R}_{\geq 0}}
\newcommand{\defeq}{\coloneq}
\newcommand{\chalg}{\textsc{ConvexHullSampler}}
\newcommand{\Sdist}{\mathcal{S}}

\newcommand{\what}{\tilde{w}}

\newcommand{\all}{A}
\newcommand{\pay}{\pi}

\newcommand{\polytope}{\mathcal{P}}
\newcommand{\util}{\textsc{Util}}

\newcommand{\alloc}{\textsc{Alloc}}
\newcommand{\price}{p}
\newcommand{\mech}{M}

\newcommand{\buymany}{\textsc{BuyMany}}
\newcommand{\itempricing}{\textsc{ItemPricing}}
\newcommand{\seqitempricing}{\textsc{Seq}\itempricing}
\newcommand{\expostitempricing}{\textsc{ExPost}\itempricing}

\DeclareMathOperator*{\pr}{\operatorname{Pr}}
\DeclareMathOperator*{\expect}{\operatorname{E}}

\newcommand{\probover}[2]{\pr_{#1}\left[ #2 \right]}

\newcommand{\xpectover}[2]{\expect_{#1}\left[ #2 \right]}
\newcommand{\bigo}[1]{O\left(#1\right)}
\newcommand{\bigoh}{O}

\DeclareMathOperator*{\argmax}{arg\,max}

\DeclareMathAlphabet{\mymathbb}{U}{BOONDOX-ds}{m}{n}

%% file: sec_intro.tex
\section{Introduction} \label{sec:intro}

In the last decade, strong connections have emerged between mechanism design, online selection, and optimal stopping problems such as prophet inequalities, leading to a number of beautiful results and techniques in both areas. Among the surprising insights to emerge from this body of work is the small gap between online and offline optimization, as well as between simple and optimal mechanisms, for the objective of social welfare maximization. Consider, in particular, a seller wishing to allocate $m$ items across $n$ buyers with combinatorial valuations over the items, denoted by $v_i$ for $i\in [n]$. 
The seller's goal is to find a partition $\{S_1, S_2, \cdots, S_n\}$ of the set $[m]$ of items and allocate set $S_i$ to each buyer $i$ so as to maximize the social welfare, defined as $\sum_{i\in [n]} v_i(S_i)$. The seller can do so using the VCG mechanism. However, from a single buyer's perspective, his allocation and payment to the seller as functions of his and others' values can be complicated and inscrutable. It turns out that the seller can use a much simpler mechanism at a small constant factor loss, when the buyers' values are structured and drawn from known distributions. For example, if buyers have fractional subadditive (i.e. XOS) values, the seller can obtain a 2-approximation via a {\em sequential item pricing}: she interacts with each buyer one at a time, offering any remaining items at pre-computed fixed prices, and allowing the buyer to select his favorite set of items to buy. 

But what if the seller wants to maximize her revenue? Can a simple online mechanism like sequential item pricing still approximate the optimal offline mechanism, and if so, under what assumptions? This work addresses and answers these questions.

In this paper, we make an explicit connection between revenue maximization for many combinatorial buyers and online contention resolution schemes (henceforth, OCRSes), a key technique in designing prophet inequalities. {\bf We develop a framework for designing OCRSes for revenue maximization over a given class of mechanisms, and instantiate it with an OCRS over sequential item pricings.} As a consequence of this construction, we obtain novel simplicity versus optimality results for revenue maximization. 
Our results hold for arbitrary distributions over subadditive valuation functions. We discuss these implications for mechanism design in detail below.

\subsubsection*{Contention Resolution Schemes and their connections to sequential mechanisms}

An OCRS is an online rounding of a fractional relaxation of the offline objective for packing-type problems. OCRSes were originally designed for maximizing set functions \citep{vondrak11submodular, feldman16online}, the natural algorithmic counterparts of ``single parameter" prophet inequalities. Let us consider the example of a seller wishing to sell a single item. Each buyer's type is described by a single value $v_i\sim\D_i$ for the item; the objective function is the expectation over the joint value distribution of the value of the buyer who receives the item. The offline optimum, a.k.a. the prophet, receives a reward of $\expect_v [\max_i v_i]$. Let us now consider a fractional relaxation where  $x_i\in [0,1]$ denotes the probability that buyer $i$ receives the item. Considering buyer $i$'s contribution alone to the objective, this contribution is maximized if the buyer receives the item when his value is in the highest $x_i$ quantiles of his distribution. Let $\sw_{x_i}(\D_i)$ denote this expected contribution. The fractional relaxation then maximizes $\max_{x\in [0,1]^n: \sum_i x_i\le 1} \sum_i \sw_{x_i}(\D_i)$. This fractional relaxation has a natural economic interpretation: we relax the {\em ex post} supply constraint (namely, that we have one item to sell) to an {\em ex ante} supply constraint (namely, that the expected number of items sold is one). Accordingly, this relaxation is called the {\em ex ante relaxation}.

An OCRS takes as input an ex ante feasible solution $x$ and ``rounds" it in an online fashion into an ex post feasible solution as follows: we first sample each $i\in [n]$ with probability $x_i$, forming a (possibly infeasible) set $R(x)\subseteq [n]$; we then pick a feasible subset of $R(x)$ in an online fashion. The goal is to show that (1) each element has a low probability of being ``blocked" by previously selected elements, and (2) conditioned on being sampled and unblocked, each element is picked with sufficiently high probability by the online algorithm. Owing to the concavity of social welfare as a function of the allocation probability\footnote{i.e., $\sw_{x}(\D_i)\ge \frac{x}{y} \sw_{y}(\D_i)$ for all $0<x<y$.}, we then obtain an approximation to the ex ante relaxation. Since the ex ante relaxation bounds the prophet's reward from above, this immediately implies a prophet inequality. Furthermore, the online rounding 
can be interpreted as a sequential posted pricing mechanism as it effectively sets a value threshold for each buyer, above which the buyer is served.

\textcite{lee18optimal} in fact showed a tight connection between prophet inequalities, OCRSes, and the ex ante relaxation: For any downwards closed feasibility constraint $\polytope$ in the single parameter setting,\footnote{Here $\polytope$ is the convex hull of the incidence vectors of all subsets of buyers that can be feasibly served.} the optimal competitive ratio achievable against the prophet's reward, which is also the worst-case gap between the optimal social welfare and that achievable by a sequential posted pricing, is {\em exactly} equal to the optimal competitive ratio achievable against the ex ante relaxation via an OCRS. 

\paragraph{A multi-dimensional OCRS for revenue.} We now turn to the revenue objective in the context of allocating many items. 
In a multi-dimensional OCRS, as in the single-dimensional case, our goal is to ``round" an ex ante feasible solution $x$ in an online/sequential fashion. Consider a setting with $n$ buyers, where $\D_i$ denotes buyer $i$'s value distribution, and $\D=\D_1\times\cdots\times\D_n$ denotes the joint value distribution across all buyers. Let $x_i$ denote a vector of probabilities where the $j$th component $x_{ij}$ is the probability of allocating item $j$ to buyer $i$. We will say that a single-buyer mechanism for buyer $i$ satisfies the ex ante constraint $x_i$ if this mechanism allocates each item $j$ with probability at most $x_{ij}$ to the buyer, where the probability is taken over the randomness in the mechanism as well as the randomness in the buyer's values. Then, we can write $\rev_{x_i}(\D_i)$ as the maximum revenue that can be obtained from buyer $i$ from any mechanism that satisfies the constraint $x_i$. The {\em ex ante relaxation for revenue} is then given by $\earev(\D) \defeq \max_{x} \sum_i \rev_{x_i}(\D_i)$, where the maximum is taken over all vectors $\{x_i\}_{i\in [n]}$ satisfying $\sum_i x_i \preceq (1, 1, \cdots, 1)$.\footnote{For vectors $y$ and $z$, we say $y\preceq z$ if for all coordinates $i$ we have $y_i\le z_i$.} Likewise, we can define $\rev_{x_i}(\D_i,\C)$ and $\earev(\D,\C)$ by optimizing over a specific class $\C$ of single-buyer mechanisms.

The OCRS generates a sequential mechanism in which the seller interacts with the buyers one at a time. At each iteration $i$, the seller offers a (potentially different) single-buyer mechanism over the as-yet-unsold items, say $S_i$, to the current buyer $i$, and buyer $i$ purchases some subset of these items. The goal is to design a single-buyer mechanism for each $i$ that extracts a sufficiently large fraction of the target revenue $\rev_{x_i}(\D_i)$ from $S_i$ while continuing to satisfy the ex ante constraint $x_i$. Of course, in order for this to be possible, $S_i$ should contain sufficiently many items. We formally define a revenue OCRS as follows.

\begin{infdef}
    An {\bf $\alpha$-OCRS for revenue maximization under a class of mechanisms $\C$} takes as input an ex ante allocation constraint $x\in [0,1]^m$, a value distribution $\D$, and a random subset of items $S\subseteq [m]$, and returns a single-buyer mechanism $M(\D,S,x)\in\C$ such that:
     \begin{enumerate}[(a)]
        \item $M$ only allocates items in $S$ to the buyer. Furthermore, for each $j\in S$, $M$ allocates $j$ to a buyer with value distribution $\D$ with probability at most $x_j$.
        \item $M$'s expected revenue from a buyer with value distribution $\D$ is at least a $1/\alpha$ fraction of $\beta\cdot\rev_{x}(\D,\C)$, where $\beta \defeq \min_{j\in [m]}\pr[j\in S]$.
    \end{enumerate}    
\end{infdef}

We show that if a class $\C$ of single-buyer mechanisms admits an $\alpha$-OCRS, then there exists a sequential $\C$ mechanism that achieves an $\bigoh(\alpha)$-approximation to $\earev(\D,\C)$. In other words, the gap between online and offline $\C$ mechanisms is bounded by $\bigoh(\alpha)$.

The existence of an OCRS depends on the properties of the class $\C$ of mechanisms being considered, as well as on the valuation functions of the buyers. Our main theorem shows the existence of an $\bigoh(\log m)$-OCRS for item pricing mechanisms and subadditive valuations. 

\begin{inftheorem}
\label{it:ocrs}
    There exists an $\bigoh(\log m)$-OCRS for revenue maximization under item pricing mechanisms when all buyers have (arbitrary distributions over) subadditive values over $m$ items and the ex ante constraint $x$ satisfies $x_j\ge 1/\operatorname{poly}(m)$ for all $j\in [m]$.
\end{inftheorem}

We remark that although multi-item, i.e. combinatorial, prophet inequalities have been studied extensively for the social welfare objective (see, e.g., \citet{doi:10.1137/1.9781611973730.10}), they do not use OCRSes and are based on other techniques. We review this literature in Section~\ref{sec:related}. On the other hand, whereas prophet inequalities designed for the social welfare objective in the single-parameter setting extend immediately to the revenue objective via virtual values, no such general connections exist between revenue and social welfare in the multi-item setting (some exceptional special cases are discussed in Section~\ref{sec:related}). This necessitates designing an OCRS specifically tailored to the revenue objective. To our knowledge, the only work prior to ours containing a multi-dimensional OCRS for revenue is the recent work of \citet{chawla23buy} that proves a $2$-OCRS for unit-demand valuations -- a special case of our result above.

\subsubsection*{Implications for revenue maximization}

Revenue maximization for buyers with multi-dimensional values is notoriously challenging. Even for settings with just one buyer and two items, the optimal revenue is inapproximable to within any factor by a simple mechanism (indeed, by any mechanism with a finite description complexity) \cite{BCKW-JET15, hart2013menu}. A recent line of work initiated by \textcite{chawla2022buy} overcomes this challenge by considering revenue maximization over so-called {\em Buy-Many Mechanisms}, that essentially restrict the seller to offering subadditive pricings over allocations to the buyer. \citeauthor{chawla2022buy} showed that the optimal buy-many revenue for a single buyer is approximable to within an $\bigoh(\log m)$ factor by item pricings for any distribution over buyer values, where $m$ is the number of items. This result is tight.

Multi-buyer settings present a further challenge in the manner that buyers impose externalities upon one another (i.e. the potential loss in revenue from allocating an item to one buyer instead of another). In single item settings, we can characterize this externality using virtual values. But in multi-item settings the unwieldy structure of multi-dimensional incentive constraints disallows such a characterization. The {\em ex ante relaxation} was first proposed by \citet{doi:10.1137/120878422} to address precisely this challenge: it allows breaking up and reducing the multi-agent problem into its single-agent counterparts and approximating each one separately. 

In a recent work, \textcite{chawla23buy} brought these two lines of work together and defined an ex ante relaxation for multi-buyer buy-many mechanisms, $\earev(\D,\buymany)$.
They showed that when every buyer has a unit-demand valuation function, the buy-many ex ante relaxation can be approximated by a sequential item pricing mechanism to within an $\bigoh(\log m)$ factor. They left open the question of designing an approximation for the buy-many ex ante relaxation for other classes of valuations. Our work resolves this open question by bounding this gap for arbitrary distributions over {\em subadditive values} to within a $\bigoh(\log^2 m)$ factor as a direct consequence of our OCRS construction.

\begin{inftheorem}
\label{it:exante}
    For settings with subadditive buyers over $m$ items, sequential item pricing obtains at least an $\bigoh(\log m)$ fraction of the ex ante item pricing revenue and at least an $\bigoh(\log^2 m)$ fraction of the optimal multi-buyer buy-many revenue.
\end{inftheorem}

We summarize our results on the approximation of the ex-ante item pricing revenue by sequential item pricings for different classes of value distributions in Table~\ref{tab:result-summary}. All of our upper bounds extend to approximation against the optimal buy-many revenue at a loss of an additional $O(\log m)$ factor via the results of \citep{chawla23buy}.

\bgroup
\def\arraystretch{1.35}
\begin{table}[ht!]
\centering
\begin{tabular}{|c|c|c|c|}
    \hline
     & Upper Bounds & Lower Bounds & Reference  \\ 
    \hline
    Unit-Demand & $2$ &  $2$  & \cite{chawla23buy} \\ 
    \hline
    Gross Substitutes &  $2$   &   $2$ & \Cref{thm:gross-subs} and \cite{chawla23buy}\\
    \hline
    XOS & $\bigoh(\log m)$  &  $\Omega\left(\log^{1/2} m\right)$ & \Cref{thm:maintheorem,thm:xos-lb}  \\
    \hline
    Subadditive & $\bigoh(\log m)$  &  $\Omega\left(\log^{1/2} m\right)$ & \Cref{thm:maintheorem,thm:xos-lb}  \\
    \hline
    Monotone &  $\min\{n, 4m^2\}$ &  $\Omega\left(\min\{n,\sqrt{m}\}\right)$ &  \Cref{thm:monotone-ub,thm:monotone-lb} \\
    \hline
\end{tabular}
\caption{Summary of known results and our contributions for approximating $\earev(\D,\itempricing)$ by $\rev(\D,\seqitempricing)$ for different families of buyer valuations.}
\label{tab:result-summary}
\end{table}
\egroup

\subsubsection*{Designing the OCRS for item pricings} 

Let us now consider the single-agent problem at the heart of our multi-dimensional OCRS. We are given a random subset $S$ of items that contains every item with (say) $1/2$ probability and an ex ante supply constraint $x$. We want to extract (say) half as much revenue from the buyer over this subset of items, as the optimal mechanism (call it $M^*$) obtains from the buyer over the entire set of items, with both our mechanism and the optimal one satisfying the given ex ante constraint $x$. Let $T^*$ denote the set of items a buyer purchases under $M^*$. If we could manage to sell $T^*\cap S$ to the buyer at the same prices as $M^*$, then we would be done. But the trouble with having only a subset of items available for sale is that the buyer may switch from wanting to buy $T^*\cap S$ to a completely different set of items, say $T$. This leads to two challenges: first $T$ may not be able to generate as much revenue as $T^*$; second, this switch may cause us to violate the ex ante supply constraint by selling some items with a higher probability than intended (e.g. the items in $T\setminus T^*$). We show how to resolve both of these issues for item pricings. 

Our first main technical contribution is to show that if the buyer has subadditive values and we are allowed to ignore the ex ante constraint $x$, for any given set $S$, we can generate revenue from the items in $S$ that is comparable (within a logarithmic factor) to the revenue $M^*$ obtains from the same set $S$ of items. We use the subadditivity of values to argue that the buyer obtains sufficiently high utility under $M^*$ from the items in $S$. Then, using an approach developed by \textcite{chawla2022buy}, we scale up the prices in $M^*$ to extract a fraction of this utility as revenue. This {\em Revenue Recovery Scheme} (RRS) (formally defined in Section~\ref{sec:prelims}) satisfies the revenue requirement of an OCRS as well as guarantees that the {\em total} number of items sold is not much larger than the number sold by $M^*$. However, it does not necessarily satisfy the given {\em per-item} allocation constraints. 

Our second technical contribution fixes the per-item  allocation constraints while losing only a constant factor in revenue. We make use of the fact that a revenue recovery scheme as described above exists for every subset $T$ of items and only allocates items in $T$. To obtain an OCRS, we carefully choose a distribution over all subsets $T$ of $S$ and apply the RRS to the chosen subset to obtain a random pricing. Suppose, for example, that the RRS applied to the entire set $S$ oversells some item $j$. Then, instead of deterministically choosing the pricing corresponding to $S$, with some appropriate probability we drop the item $j$, apply the RRS to $S\setminus\{ j\}$, and choose the corresponding pricing. Having to drop items in this manner hurts our revenue guarantee. But importantly, the RRS cannot oversell too many items as it satisfies the ex ante constraint in aggregate over all items. Therefore it becomes possible to choose a random set $T$ such that the total revenue contribution over $T$ is large enough, while at the same time, the per-item allocation probabilities match those of the given ex ante constraint. We call this procedure the {\em Convex Hull Sampler} as it produces a random pricing in the convex hull of those given by the RRS at different sets $T\subseteq S$.

Unfortunately, the multiplicative loss in revenue for the RRS we design depends on the ratio of the maximum to the minimum price charged in $M^*$, and this dependence is necessary. This dependence can arise due to a potential long tail of the buyer's value distribution leading to exponentially large prices coupled with exponentially small allocation probabilities. Our final technical contribution is to eliminate this dependence in our approximation of ex ante revenue. In particular, as long as each ex ante constraint $x_j$ is at least $\Omega(1/\operatorname{poly}m)$, it becomes possible to obtain an $\bigoh(\log m)$-OCRS for item pricings. This in turn provides the $\bigoh(\log m)$ approximation to the ex ante revenue for item pricings stated in \Cref{it:exante}.

Finally, we remark that our OCRS and sequential item pricing are fully constructive and can be found in time polynomial in $n$, $m$, and the support sizes of the distributions $\D_i$, assuming: (1) We have access to the ex ante optimal solution $x$ and the corresponding optimal (random) pricing $p$; (2) We have access to a demand oracle that given an item pricing and a value function in the support of $\cup_i\D_i$ returns the set of items bought by the buyer with value $v$.

\subsubsection*{Overview of the rest of the paper}

In \Cref{sec:prelims} we formally introduce the ex ante benchmark, the OCRS for revenue maximization, and an outline of our approach towards designing an OCRS for item pricings. We present the four main parts of our upper bound construction in \Cref{sec:rrs-subadd,sec:rrs_to_ocrs,sec:OCRStoEARev,sec:subadd-theorem}. All our constructions are presented as existential results; we address the computational aspects in \Cref{sec:computational-aspects}.  In \Cref{sec:xos-lb} we provide an $\Omega(\sqrt{\log m})$  lower bound on the existence of an OCRS for item pricing over the class of XOS valuation functions. Our lower bound in fact applies to the gap between the ex ante item pricing relaxation and the revenue of {\em any} ex post feasible item pricing mechanism. We defer a discussion of gross substitutes and general valuations to the appendix. \Cref{sec:related} discusses related work. We conclude and outline directions for future work in \Cref{sec:conclusion}.  

%% file: sec_prelims.tex
\section{Definitions and an Outline of Our Construction} \label{sec:prelims}

We consider the standard mechanism design setting with $m$ items and $n$ buyers with the objective of maximizing the total revenue of the seller. Buyers have combinatorial valuations over the items, $v_i:2^{[m]}\rightarrow \Rplus$ for $i\in [n]$, drawn from known independent distributions $\D_i$. We write $\D=\D_1\times\cdots\times\D_n$ as the joint value distribution. Section~\ref{sec:PrelimMD} of the appendix describes different classes of valuations and Bayesian incentive compatible (IC) mechanisms for readers unfamiliar with mechanism design. Henceforth, we assume knowledge of these concepts. 

For a buyer with valuation $v$ and a single-buyer mechanism represented as a (random) pricing $p$ over lotteries, we write $\util(v,p)$ as the utility of the buyer; $\alloc(v,p)$ as the allocation made to the buyer (in the form of a random subset of items or an indicator vector for that subset); and $\rev(v,p)$ as the corresponding revenue of the mechanism. We extend these definitions to a value distribution $\D$ by taking expectations over $v\sim\D$, and to a class of mechanisms $\C$ by taking the maximum over all $p\in \C$ of the expected revenue of $p$. For $j\in [m]$, a subscript of $j$ on each of these quantities indicates the contribution of item $j$ to the corresponding quantity, where well-defined. Finally, for any valuation function $v$ and a subset  $S\subseteq [m]$ of items, we denote by $v_{|S}$ the valuation function given by $v_{|S}(T)=v(T\cap S)$ for all $T\subseteq [m]$; analogously, we denote by $\D_{|S}$ the value distribution that first samples a valuation $v\sim \D$, and then returns $v_{|S}$. In other words, $\D_{|S}$ captures the valuation of buyer $\D$ if we can only offer them items in $S$.

\subsection{The Ex Ante Relaxation} 

As discussed previously, the ex ante relaxation provides an upper bound on the revenue of an optimal mechanism by relaxing the ex post supply constraint, namely that each item should be sold at most once, to an ex ante supply constraint, namely that the expected number of copies of the item sold is at most one. 

Following the approach of \textcite{chawla23buy}, we can define an ex ante relaxation with respect to a specific class of mechanisms. Let $\C$ be a class of single-buyer Bayesian IC mechanisms, and let $\M$ be a distribution over mechanisms in this class. For a buyer with value distribution $\D$, and an ex ante constraint $x\in [0,1]^m$, we say that the distribution $\M$ satisfies the ex ante constraint $x$ over $\D$ if it holds that:
\[
    \xpectover{M\sim\M}{\alloc(\D,M)} \preceq x.
\]

\noindent
The optimal revenue achievable from a class $\C$ of mechanisms given ex ante constraint $x$ is:
\[\rev_{x}(\D,\C) \defeq \max_{\M\in \Delta^{\C}: \M \text{ satisfies } x \text{ over } \D} \xpectover{M\sim \M} {\rev(\D,M)}.\]

\noindent
The optimal ex ante revenue from $n$ buyers with joint value distribution $\D$ is then defined as follows.
\begin{definition}[Optimal Ex Ante Revenue under Mechanism Class] Given a class $\C$ of single-buyer mechanisms and a joint value distribution $\D$, we define the \emph{optimal ex ante revenue under} $\C$ as
\[
    \earev(\D,\C) \defeq \max_{x\in\polytope}\sum_{i=1}^n\rev_{x_i}(\D_i,\C),
\]
where $\polytope$ denotes the polytope of feasible ex ante allocations; that is
\[
    \polytope \defeq \left\{x = (x_1,x_2,\dotsc , x_n): x_i\in\Rplus^m\;\forall i\in [n]\text{ and }\sum_{i=1}^nx_{ij} \leq 1\;\forall j\in [m]\right\}.
\]
\end{definition}

Note that the ex ante revenue $\earev(\D,\C)$ provides an upper bound on the revenue of any Bayesian IC mechanism $M=(\all,\pay)$, each single-agent component of which is a $\C$ mechanism; that is, for all $i\in [n]$, the distribution $(\all(v_i,\cdot), \pay(v_i,\cdot))$ taken over the randomness in $v_{-i}$ lies in $\Delta^{\C}$.

\textcite{chawla23buy} prove the following connection between the ex ante revenue over buy many mechanisms and the ex ante item pricing revenue.

\begin{theorem}[\cite{chawla23buy}]
\label{t:buy-many-item-pricing-approx}
    For any multi-buyer value distribution $\D$ over $m$ items, \[\earev(\D, \buymany)\le 
\bigoh(\log m)\cdot \earev(\D, \itempricing).\]
\end{theorem}

\subsection{Our Main Results}

Given the connection between the ex ante buy-many and item pricing revenues established in~\Cref{t:buy-many-item-pricing-approx}, we focus on approximating the ex ante relaxation for item pricings. Our main result is as follows:

\begin{restatable}{theorem}{maintheorem} \label{thm:maintheorem}
    Let $\D$ be any joint distribution for $n$ buyers and $m$ items over subadditive valuation functions, and let $\seqitempricing$ denote the class of all Sequential Item Pricing mechanisms. Then,
    \[
        \earev(\D,\itempricing) \leq \bigoh(\log m) \cdot \rev(\D, \seqitempricing).
    \]
\end{restatable}

As a corollary, we obtain the following approximation to ex ante buy-many revenue. 

\begin{corollary}
    Let $\D$ be any joint distribution for $n$ buyers and $m$ items over subadditive valuation functions. Then we have,
    \[\earev(\D,\buymany)\le \bigoh(\log^2 m)\cdot \rev(\D,\seqitempricing).\]
\end{corollary}

We show that the dependence on $m$ in Theorem~\ref{thm:maintheorem} is necessary for subadditive valuations: there is an $\Omega(\sqrt{\log m})$ gap between ex ante item pricing revenue and the revenue of sequential item pricings even when buyer valuations are XOS (fractionally subadditive). In fact, the lower bound also applies to the gap between the ex ante item pricing revenue and the revenue of {\em any} ex post feasible (random) item pricing mechanism. We prove this lower bound in \Cref{sec:xos-lb}. 

\begin{restatable}{theorem}{xoslbthm}\label{thm:xos-lb}
    There exists a joint value distribution $\D$ for $n$ buyers and $m$ items over fractionally subadditive (XOS) valuation functions for which
    \begin{align*}
        \earev(\D,\itempricing) &\geq \Omega\left(\sqrt{\log m}\right) \cdot \rev(\D, \seqitempricing).
    \end{align*}
\end{restatable}

We then consider other classes of value distributions. \textcite{chawla23buy} previously showed that sequential item pricings obtain a constant fraction of the ex ante item pricing revenue over the class of unit demand value distributions. We show an extension of that result to gross substitutes valuations in \Cref{sec:gs-improvement-main}. Note that the factor of $2$ is tight even for $m=1$ items.

\begin{restatable}{theorem}{gstheorem}
\label{thm:gross-subs}
    Let $\D$ be any joint distribution for $n$ buyers and $m$ items over gross substitutes valuation functions. Then we have,
    \[\earev(\D,\itempricing) \leq 2\cdot\rev(\D,\seqitempricing) .\]
\end{restatable}

Finally, in \Cref{sec:monotone} we demonstrate upper and lower bounds for the family of general monotone valuations. Our results are summarized in \Cref{tab:result-summary}.

\subsection{Online Contention Resolution Schemes for Revenue}

Our main goal in this paper is to approximate the ex ante revenue for a given class $\C$ of mechanisms using a sequential $\C$ mechanism. To this end, we define the following \textit{Online Contention Resolution Scheme} (OCRS).

\begin{definition}[OCRS for Revenue Maximization]
\label{def:OCRS}
    Let $\F$ be a family of single-buyer value distributions over $m$ items and $\alpha\ge 1$.    
    We say that $\F$ admits an $\alpha$-OCRS for revenue maximization under a class of mechanisms $\C$ if, for any allocation constraint $x\in [0,1]^m$, any distribution $\D\in\F$ and any random subset of items $S\subseteq [m]$, there exists a (possibly random) pricing menu $q=q(\D,S,x)\in\C$ such that:
     \begin{enumerate}[(a)]
        \item $\xpectover{S}{\alloc(\D_{|S},q)}\preceq x$, and,
        \item $\xpectover{S}{\rev(\D_{|S},q)} \geq \frac{1}{\alpha} \cdot \beta\cdot\rev_{x}(\D,\C)$, where $\beta \defeq \min_{j\in [m]}\pr[j\in S]$.
    \end{enumerate}
    When $\C$ is the class of item pricings, we replace (b) with a stronger condition, namely
    \begin{enumerate}[(a')]
    \setcounter{enumi}{1}
        \item $\xpectover{S}{\rev(\D_{|S},q)} \geq \frac{1}{\alpha}\cdot \xpectover{S, p}{\sum_{j\in S}p_j\alloc_j(\D,p)}$, 
    \end{enumerate}    
    where $p$ is the (random) item pricing that achieves $\rev_{x}(\D,\itempricing)$.
\end{definition}

The following lemma connects OCRSes for revenue to approximations for ex ante revenue. We present its proof in Section~\ref{sec:OCRStoEARev}.

\begin{restatable}{lemma}{ocrstoapproxlemma}\label{t:ocrs_argument}
    If a family of value distributions $\F$ admits an $\alpha$-OCRS for revenue maximization under a class of mechanisms $\C$, then for any $\D\in\F$, we can obtain a $4\alpha$-approximation to $\earev (\D,\C)$ using a sequential $\C$ mechanism.
\end{restatable}

\subsection{Establishing OCRSes through Revenue Recovery Schemes}

As we saw, in order to design an OCRS for revenue maximization, we need to be able to satisfy an allocation condition as well as a revenue condition over the given set of available items. Our first main technical contribution is to show that for item pricing mechanisms we can ignore the allocation condition and instead simply require that the pricing we construct is component-wise not much smaller than the ex ante-optimal pricing.
To that end, we introduce the concept of a \textit{Revenue Recovery Scheme} (RRS).

\begin{definition}[Revenue Recovery Scheme]
\label{def:RRS}
Let $\F$ be a family of single-buyer value distributions over $m$ items and $\alpha\ge 1$. We say that $\F$ admits an $\alpha$-RRS if for any $\D\in\F$, any item pricing $p$ and any subset of items $S\subseteq [m]$, there exists an item pricing $q=q(\D,S,p)$ such that:
    \begin{enumerate}[(a)]
        \item $q_j\geq \frac{1}{\alpha}\cdot p_j$ for any $j\in S$, and,
        \item $\rev(\D_{|S}, q) \geq \frac{1}{\alpha}\cdot \sum_{j\in S}p_j\alloc_j(\D,p)$.
    \end{enumerate}
\end{definition}

Perhaps surprisingly, we show in \Cref{sec:rrs_to_ocrs}, that the lower bounds on prices enforced in condition (a) are enough to guarantee an OCRS for the class of item pricing mechanisms, even though no condition on the allocation is explicitly imposed.

\begin{lemma}\label{t:rrs_argument}
    Any family of value distributions $\F$ that admits an $\alpha$-RRS also admits an $\frac{e}{e-1}\alpha$-OCRS for revenue maximization under $\itempricing$.
\end{lemma}

It remains to argue that good revenue recovery schemes exist for valuation functions of interest. As a warm-up, we observe that gross substitutes valuations admit a $1$-RRS. Indeed, consider the RRS that sets $q\defeq p$. Clearly, the first scaling condition holds for $\alpha = 1$. Furthermore, we can interpret the offering of prices $p$ over a subset $S$ of items as being equivalent to offering the prices $q_j=p_j$ for $j\in S$ and $q_j=\infty$ for $j\not\in S$. Then, by the definition of gross substitutes valuations, a buyer with value $v$ that purchases some item $j\in S$ under $p$ continues to purchase $j$ under $q$, and the condition (b) follows. This is summarized in the following proposition. Combined with Lemma~\ref{t:ocrs_argument} this immediately implies a constant approximation to the ex ante revenue for gross substitute valuations. We show how to improve the approximation to $2$ and prove Theorem~\ref{thm:gross-subs} in Appendix~\ref{sec:gs-improvement-main}.

\begin{proposition} \label{prop:GS-RRS-result}
    Gross substitutes valuations admit a $1$-RRS and an $\frac{e}{e-1}$-OCRS.
\end{proposition}

\noindent

Our main result for revenue recovery schemes applies to the much more general class of subadditive valuations and is stated through the following lemma that we prove in Section~\ref{sec:rrs-subadd}.
\begin{lemma} \label{thm:subadd-RRS}
    The family of subadditive valuations admits an $\bigoh(\log m + \log \Gamma)$-RRS, where $\Gamma \defeq p_{max}/p_{min}$ is the aspect ratio of the item pricing $p$ to be revenue-recovered.
\end{lemma}

Observe that Theorem~\ref{thm:maintheorem} does not immediately follow by combining \Cref{t:ocrs_argument}, \Cref{t:rrs_argument}, and \Cref{thm:subadd-RRS} because of the dependence on the aspect ratio $\Gamma$ in \Cref{thm:subadd-RRS}. This dependence is in fact intrinsic to the RRS definition, as there are XOS instances where it is impossible to obtain an $o(\sqrt{m})$-RRS---a more detailed discussion can be found in~\Cref{sec:rrs-lb}. We show how to remove this dependence for Theorem~\ref{thm:maintheorem} in~\Cref{sec:subadd-theorem} via an end-to-end argument. 

%% file: sec_subadd.tex
\section{An RRS for Subadditive Valuations (Proof of~\Cref{thm:subadd-RRS})}
\label{sec:rrs-subadd}

In this section, we design a \textit{Revenue Recovery Scheme} for \textit{subadditive} valuations, proving~\Cref{thm:subadd-RRS}. Fix any item pricing $p\in\Rplus^m$ and any set of items $S\subseteq [m]$. Let $\Gamma'\defeq (\max_{j\in S}p_j)/(\min_{j\in S}p_j)$, where the minimum excludes any $j$ such that $p_j=0$. Although the definition of an RRS calls for a deterministic pricing $q$, we first show the existence of a random pricing that satisfies the conditions of the scheme.

Our algorithm
samples a real $\gamma\in [\ell, h]$ from the distribution with density $f(\gamma) = 1/(\gamma\cdot\log\frac{h}{\ell})$ for $\ell = 1/2$ and $h=m\Gamma'$ and then returns a \textit{uniform} scaling of $p$ by $\gamma$, i.e. $q\defeq\gamma  p$. Clearly, any price is scaled down by at most $1/2$, so the scaling RRS condition holds. In order to prove~\Cref{thm:subadd-RRS}, we will now show that for any subadditive buyer $v$,

\begin{equation}\label{eq:subadd-rrs-to-show}
    \rev(v_{|S}, q) \geq \frac{1}{2\log (2m\Gamma')}\cdot \sum_{j\in S}p_j \alloc_j(v,p).
\end{equation}

As a first step towards proving this inequality, we employ Lemma 3.1 from \cite{chawla2022buy} that connects the revenue from any buyer with the buyer's utility under a uniform scaling of an arbitrary pricing. 

\begin{fact}[Lemma 3.1 of \cite{chawla2022buy}]
For any monotone valuation function $v$, any pricing $p$ and any $0<\ell\leq h$, if $\gamma \in [\ell, h]$ is drawn from the distribution with density $f(\gamma) = 1/(\gamma\cdot\log\frac{h}{\ell})$, then
\[\operatorname{E}_{\gamma}[\rev(v, \gamma p)] \geq \frac{1}{\log (h/l)}\cdot \bigg( \util(v,\ell p) - \util(v,hp)\bigg).\]
\end{fact}

\noindent
By applying this lemma to $v_{|S}$ and $q$, we immediately obtain that
\begin{equation}\label{eq:rev-utility-lemma-application}
\rev(v_{|S}, q) \geq \frac{1}{\log (2m\Gamma')}\cdot \bigg( \util(v_{|S}, p/2) - \util(v_{|S},m\Gamma' p)\bigg).
\end{equation}

We now proceed to bound the difference in the utility terms. Observe that as a function of the scaling $\gamma$,
$\util(v_{|S}, \gamma p)\defeq \max_{T\subseteq S} ( v(T) - \gamma p(T))$
is non-negative and non-increasing. Furthermore, as long as it remains strictly positive, it decreases at a rate of $p(T)$ for some $T\subseteq S$, which is at least $\min_{j\in S}p_j$. Thus, assuming that $\util(v_{|S},h p) > 0$, we obtain that
\[ \util(v_{|S},\ell p) - \util(v_{|S},h p) \geq (h-\ell)\cdot \min_{j\in S}p_j \geq \frac{m}{2}\cdot \max_{j\in S}p_j \geq \frac{1}{2}\cdot\sum_{j\in S} p_j \alloc_j(v_{|S}, p)\]
under our definition of $\ell$ and $h$.

It remains to bound the utility drop if $\util(v_{|S},h p) = 0$; this is precisely where the subadditivity of values is required. Let $A^*\defeq\alloc(v,p)$. Then, by definition we have that
\begin{align*}
    \util(v_{|S}, \ell p) \defeq \max_{T\subseteq S} \bigg( v(T) - \frac{1}{2}p(T)\bigg)
    &\geq v(A^*\cap S) - \frac{1}{2}p(A^*\cap S) \\
    &\geq v(A^*) - v(A^*\setminus S) - \frac{1}{2}p(A^*\cap S) \\
    &= (v(A^*)-p(A^*)) - (v(A^*\setminus S) - p(A^*\setminus S)) + \frac{1}{2}p(A^*\cap S) \\
    &\geq \frac{1}{2}p(A^*\cap S) = \frac{1}{2}\cdot \sum_{j\in S}p_j\cdot \alloc_j(v,p),
\end{align*}
where the second line uses the subadditivity of $v$, the third line uses the additivity of prices and the last one uses the fact that $A^*$ is the utility maximizing set of items for $v$ under pricing $p$.

Thus, we have shown that in either case
\begin{equation}\label{eq:utility-drop-bound}
    \util(v_{|S},\ell p) - \util(v_{|S},h p) \geq  \frac{1}{2}\cdot\sum_{j\in S} p_j \alloc_j(v_{|S}, p),
\end{equation}
and combining~\eqref{eq:utility-drop-bound} with~\eqref{eq:rev-utility-lemma-application} we obtain~\eqref{eq:subadd-rrs-to-show} as desired. Finally, by taking the expectation of~\eqref{eq:subadd-rrs-to-show} over any distribution $\D$, we obtain
\begin{equation*}
    \operatorname{E}_q\big[ \rev(\D_{|S}, q) \big]\geq \frac{1}{2\log (2m\Gamma')}\cdot \sum_{j\in S}p_j \alloc_j(\D,p).
\end{equation*}
and thus there must exist some \textit{deterministic} pricing $q$ (i.e. some \textit{deterministic} scaling $\gamma$) for which the revenue condition holds. The proof is then completed by observing that $\Gamma'\le \Gamma$.

%% file: sec_rrs_argument.tex
\section{Revenue Recovery Implies OCRS (Proof of~\Cref{t:rrs_argument})}
\label{sec:rrs_to_ocrs}

In this section we prove~\Cref{t:rrs_argument}, stating that any family of valuations that admits a $\alpha$-RRS will also admit an $(e/(e-1))\alpha$-OCRS for revenue maximization under $\itempricing$. Let $(\D, S, x)$ be the given instance of the OCRS. Let $p$ denote the random pricing that achieves the revenue bound $\rev_x(\D,\itempricing)$. Recall that to design an OCRS, our goal is to find a pricing over $S$ that will recover a large fraction of $\rev_x(\D,\itempricing)$ while respecting the ex ante constraint $x$. 

To develop some intuition for this result, let us suppose that $S$ and $p$ are deterministic, and that the pricing $p$ is uniform -- say $p_j=1$ for all $j$. Applying the RRS to $(p, S)$ provides us with a pricing $q$ with $q_j\ge 1/\alpha$ for all $j$ that achieves good revenue. Since the prices $q_j$ are comparable to the prices $p_j$, and the revenue of $q$ is no more than the revenue of $p$, we may conclude that the total number of items allocated by $q$ is not much larger than the total number allocated by $p$. In other words, the allocation constraint is satisfied \textit{in aggregate}. In fact, we can achieve the same property over any {\em subset} $T$ of $S$ by applying the RRS to $p$ over $T$ to obtain pricing $q_T$. The aggregate allocation of each $q_T$ is comparable to the allocation of $p$ over the set $T$. We want to argue that there is a distribution over the $q_T$'s, corresponding to an expected allocation vector that is small in each coordinate. At the same time, we want the aggregate expected allocation of the random pricing to be comparable to that of $p$, so as to get back a revenue guarantee. 

The key technical ingredient in our proof that enables the above properties is a novel algorithmic process called $\chalg$ (Algorithm~\ref{alg:convex-hull}). This process takes as input a target $k$-dimensional vector $w\in\Rplus^k$ and a collection of $2^k$ vectors $y^T\in\Rplus^k$ corresponding to subsets $T\subseteq [k]$, with $y^T_j=0$ for all $T$ and $j\not\in T$. The goal is to find a vector in the convex hull of the $\{y^T\}$'s that is component-wise dominated by $w$ but at the same time comparable to $w$ in length.
We show that if each $y^T$ is long enough, i.e. $|y^T|\ge \sum_{j\in T} w_j$, then we can always find such a vector $z$ with $|z|/|w|\ge 1-1/e$.\footnote{Here $|z|\defeq\sum_{j\le k} z_j$ denotes the $\ell_1$ length of the vector $z$.} 

We summarize this property of the convex hull sampler in the lemma below. 

\begin{lemma}\label{l:convex_hull_sampler_new}
Let $w\in\Rplus^k$ and $y^T\in\Rplus^k$ for all $T\subseteq [k]$ satisfy (i) $y^T_j=0$ for all $T$ and $j\not\in T$, and, (ii) $|y^T|\ge \sum_{j\in T} w_j$. Then, $\chalg(k, w, \{y^T\}_{T\subseteq [k]})$ returns a distribution $\{\lambda_T\}_{T\subseteq [k]}$ such that:
\begin{enumerate}[(a)]
    \item $\sum_T\lambda_T y^T\preceq w$, and,
    \item$|\sum_T\lambda_T y^T|\ge \frac{e-1}{e}\cdot |w|$.
\end{enumerate}
\end{lemma}

\begin{algorithm}[h!]
\caption{Algorithm $\chalg(k,w,\{y^T\}_{T\subseteq [k]})$}
\label{alg:convex-hull}
\begin{algorithmic}[1]
    \State Initialize variables $\what\gets w$, $Q\gets [k]$, $\sigma\gets 0$ and $\lambda_T\gets 0$ for all $T\subseteq [k]$.
    \While{$Q \neq \emptyset$ and $\sigma < 1$}        
        \State Let $\tau$ be the largest number for which $\tau\cdot y^{Q}\preceq \what$, i.e. $\tau\gets \min_{j\in Q} (\what_j / y^Q_j)$.
        \State Update $\lambda_Q \gets \min\{\tau, 1 - \sigma\}$, $\sigma \gets \sigma + \lambda_Q$ and $\what \gets \what - \lambda_Q\cdot  y^Q$.
        \smallskip
        \State Update $Q \gets \{i\in [k] : \what_i > 0\}$.
    \EndWhile
    \State Set $\lambda_\emptyset \gets 1 - \sigma$.
    \State Return the distribution $\{\lambda_T\}_{T\subseteq [k]}$.
\end{algorithmic}
\end{algorithm}

Before we prove the lemma, let us see how it leads to a proof of \Cref{t:rrs_argument}. Our preceding discussion suggests that we should apply the Convex Hull Sampler to the allocation vectors of the pricings $p$ and $\{q_T\}$. However, since the prices for different items can differ considerably, and in order to relate the total revenue of the pricings to the $\ell_1$ lengths of the corresponding vectors, we apply the procedure to the revenue contribution of each item.

\begin{proof}[Proof of~\Cref{t:rrs_argument}]
    We instantiate a pricing $p$ from the distribution that generates $\rev_x(\D,\itempricing)$, and instantiate a set $S$ from its distribution. For $T\subseteq S$ let $q^T$ be the deterministic pricing returned by the $\alpha$-RRS on input $(\D,T,p)$.

    Define $w_j\defeq p_j \alloc_j(\D,p)$ for $j\in S$, and $y^T_j\defeq \alpha\cdot q^T_j\alloc_j(\D_{|T}, q^T)$ for each $T\subseteq S$ and $j\in S$. We first observe that the vectors $w$ and $y^T$ satisfy the assumptions in \Cref{l:convex_hull_sampler_new}. Assumption (i) holds by construction, and (ii) follows by noting:
    \[\sum_{j\in T} y^T_j = \sum_{j\in T} \alpha\cdot  q^T_j\alloc_j(\D_{|T}, q^T) = \alpha\cdot\rev(\D_{|T},q^T)\ge \sum_{j\in T} p_j\alloc_j(\D,p) = \sum_{j\in T} w_j.\]
    Here the inequality follows from property (b) of the $\alpha$-RRS (Definition~\ref{def:RRS}).

    Now, let $\lambda \defeq \{\lambda_T\}_{T\subseteq [k]}$ be the distribution returned by $\chalg$ on $(S, w, \{y^T\}_{T\subseteq S})$. Consider the following random pricing: we pick a random set $T$ from the distribution $\lambda$ and return the random pricing $\tilde{q}\defeq q^T$.

    Let us first analyze the expected allocation of this random pricing. We first note that for all $T$, we have $y^T_j=\alpha\cdot q^T_j\alloc_j(\D_{|T}, q^T)\ge p_j\cdot \alloc_j(\D_{|T}, q^T)$ from property (a) of the $\alpha$-RRS (Definition~\ref{def:RRS}). The expected allocation of the random pricing $\tilde{q}$ is then given by:
    \begin{align*}
    \alloc_j(\D_{|S}, \tilde{q}) = \sum_{T\subseteq S}\lambda_T \cdot \alloc_j(\D_{|T},q^T) \leq  \sum_{T\subseteq S}\lambda_T \cdot \frac{y^T_j}{p_j} \leq \frac{w_j}{p_j} \leq \alloc_j(\D, p),
    \end{align*}
    where the second to last inequality follows from property (a) in \Cref{l:convex_hull_sampler_new} and the last inequality follows from the definition of $w$.

    Next let us consider the revenue obtained by $\tilde{q}$. By the definition of the $y^T$'s, we have:
\begin{align*}
    \rev(\D_{|S}, \tilde{q})  = \frac{1}{\alpha}\cdot \sum_{T\subseteq S}\lambda_T |y^T| \geq \frac{1}{\alpha}\cdot \left(1-\frac{1}{e}\right) \cdot |w| = \frac{1}{\alpha}\cdot \left(1-\frac{1}{e}\right) \cdot \sum_{j\in S}p_j \alloc_j(\D,p).
\end{align*}
The lemma then follows by taking expectations for both the allocation and the revenue bounds over the randomness of $S$ and $p$.
\end{proof}

We conclude this section by presenting the proof of \Cref{l:convex_hull_sampler_new}.

\begin{proof}[Proof of \Cref{l:convex_hull_sampler_new}]
    We first observe by construction that $\lambda\defeq \{\lambda_T\}_{T\subseteq [k]}$ forms a proper probability distribution. Also, it follows by construction that $\sum_T \lambda_T y^T\preceq w$. This is because we track the vector $w-\sum_T \lambda_T y^T$ in the form of vector $\what$ in the algorithm, and our choice of $\lambda_Q$ ensures that no coordinate of this vector ever becomes negative. It remains to prove property (b).

    Suppose the algorithm runs for $s$ steps. Let $\what_1, \what_2, \dots, \what_s$ and $Q_1, Q_2, \dots, Q_s$ be the vector $\what$ and set $Q$ after each step of the algorithm, and conventionally define that $\what_0 = w$ and $Q_0 = [k]$. We first observe that
\[\left|\sum_{T}\lambda_T\cdot y^T\right| = \left|\sum_{i=0}^{s-1}\lambda_{Q_i}\cdot y^{Q_{i}}\right| = \left|w-\what_s\right| = \sum_{i=1}^{s}\left|\what_{i-1} - \what_i\right|,\]
since $\lambda_T = 0$ for any (non-empty) $T$ not realized in some iteration, and $w = \what_0 \succ \what_1 \succ \ldots \succ \what_s \succ 0$. If $\sigma < 1$ after the algorithm terminates, we necessarily have $Q_s=\emptyset$, which means that $\what_s=0$ and thus $|\sum_T \lambda_T \cdot y^T| = |w|$, which is clearly enough for property (b) to hold. 

Thus, it remains to argue about the $\sigma=1$ case, under which we know that $\sum_{i=0}^{s-1}\lambda_{Q_{i}} = 1$. For notational convenience, define the constants 
\[\gamma_i = \frac{|\what_{i - 1} - \what_i|}{|\what_{i - 1}|},\;\; i\in [s],\]
and observe that since $w = \what_0 \succ \what_1 \succ \ldots \succ \what_s \succ 0$, we have that $|\what_i| = \prod_{j=1}^i(1-\gamma_j)\cdot |w|$ and thus:

\begin{align*}
    \left|\sum_{T}\lambda_T\cdot y^T\right| = \sum_{i=1}^s\gamma_i\cdot |\what_{i - 1}| = |w|\cdot \sum_{i=1}^s \bigg(\gamma_i\cdot \prod_{j=1}^{i-1}(1-\gamma_j)\bigg) = |w|\cdot \bigg(1-\prod_{i=1}^s (1-\gamma_i)\bigg),
\end{align*}
where the last equality follows from an expansion of the sum of products. Up next, we will argue that $\gamma_i \geq \lambda_{Q_{i-1}}$ or equivalently that $|y^{Q_{i-1}}| \geq |\what_{i-1}|$ as $|\what_{i-1}-\what_i| = \lambda_{Q_{i-1}}\cdot |y^{Q_{i-1}}|$. This can be easily shown to be true for our given definitions of these vectors, as 
\[|y^{Q_{i - 1}}|  \geq  \sum_{j \in Q_{i - 1}} w_j \ge \sum_{j \in Q_{i - 1}} \what_{i - 1,j} = |\what_{i - 1}|,\]
where the first inequality follows from the definition of $w$ and $y^Q$ as well as the scaling guarantee of the RRS, the second inequality holds from $w=\what_0\succeq \what_{i-1}$ and the final equality holds from the fact that all the coordinates of $\what_{i-1}$ not in $Q_{i-1}$ are $0$. Equipped with $\gamma_i \geq \lambda_{Q_{i-1}}$, we can finally show property (b) by observing that 
\begin{align*}
    \left|\sum_{T}\lambda_T\cdot y^T\right| &\geq |w|\cdot \bigg(1-\prod_{i=1}^s (1-\lambda_{Q_{i-1}})\bigg) \\
    &\geq 
    |w|\cdot \bigg(1-\bigg(\frac{\sum_{i=1}^s(1-\lambda_{Q_{i-1}})}{s}\bigg)^s\bigg) \\
    &= |w|\cdot \bigg(1-\bigg(1-\frac{1}{s}\bigg)^s\bigg). 
\end{align*}
Here, the second inequality is an application of the AM-GM inequality, and the equality follows from $\sum_{i=0}^{s-1}\lambda_{Q_{i}} = 1$.
\end{proof}

%% file: sec_ocrs_argument.tex
\section{OCRS Approximates Ex Ante Revenue (Proof of~\Cref{t:ocrs_argument})}
\label{sec:OCRStoEARev}

In this section, we prove~\Cref{t:ocrs_argument}, which is restated for convenience. We remark that the argument we employ is standard and is used, e.g., to prove a competitive ratio of $4$ for the single item prophet inequality in \cite{chawla2010multi}. We reproduce it here for completeness. An informed reader can safely skip this section.

\ocrstoapproxlemma*

\noindent
Consider some $\D=\D_1\times\cdots\times\D_n$ in $\F$ and recall that
\[\earev(\D,\C) = \sum_{i=1}^n\earev_{x^*_i}(\D_i,\C)\]
for some ex-ante constraint $\{x^*_i\}_{i=1}^n$ with $\sum_{i=1}^nx^*_{ij}\leq 1$ for all items $j\in [m]$. We will now describe a sequential $\C$ mechanism; call it $M$. Let $S_i$ denote the set of available (i.e. unsold) items when buyer $i$ arrives.

\begin{enumerate}
    \item Initially, $S_1 = [m]$.
    \item When buyer $i\in [n]$ arrives:
    \begin{enumerate}
        \item Present them with the pricing menu $q_i\in\C$ that is produced from the $\alpha$-OCRS on distribution $\D_i$, subset of items $S_i$, and allocation constraint $\frac{1}{2}x^*_i$.

        \item Let $B_i\subseteq S_i$ be the set of items that the buyer purchases under this pricing; update $S_{i+1}=S_i\setminus B_i$.
    \end{enumerate}
\end{enumerate}

Let $\rev(\D,M)$ denote the expected revenue of the above mechanism. Note that the mechanism clearly satisfies the ex post supply constraint, as only available items are ever sold. The lemma then follows from the following sequence of steps bounding its revenue.

\begin{enumerate}[(i)]
    \item We have \[\rev(\D,M) = \sum_{i=1}^n\operatorname{E}_{S_i}\big[\rev(\D_{i|S_i}, q_i)\big] \geq (1/\alpha)\cdot\sum_{i=1}^n \beta_i\cdot \rev_{\frac{1}{2}x^*_i}(\D_i,\C),\] where $\beta_i = \min_{j\in [m]}\pr[j\in S_i]$. Here the first equality is from the definition of $M$ and the second inequality follows from the revenue condition of an $\alpha$-OCRS.
    \item For each $i$ and $\D_i$, we have $\rev_{\frac{1}{2}x^*_i}(\D_i,\C)\ge \frac 12 \cdot\rev_{x^*_i}(\D_i,\C)$, because one way to satisfy the ex ante constraint $\frac{1}{2}x^*_i$ is to flip a coin and use the mechanism that satisfies the ex ante constraint $x^*_i$ with probability $1/2$ and set all prices to $\infty$ with probability $1/2$.
    \item For any item $j\in [m]$ and any buyer $i\in [n]$, it holds that $\pr[j\in S_i]\geq \frac{1}{2}$. That is, $\beta_i\ge \frac{1}{2}$ for all $i$.
\end{enumerate}

\noindent
It remains to prove the claim in (iii), for which the following sequence of inequalities suffices:
\[\Pr[j \in S_i] = 1 - \sum_{k < i} \Pr[j \in B_k \mid j \in S_k]\ge 1-\sum_{k<i} \frac{1}{2} x^*_{k, j} \ge \frac{1}{2}.\]

\noindent
Here, the first inequality follows from the allocation condition of the OCRS (instantiated on constraints $\frac{1}{2}x^*_i$) and the second follows from the ex ante constraint. This completes the proof of \Cref{t:ocrs_argument}.

%% file: sec_removing_pmax_pmin.tex
\section{An $\bigoh(\log m)$-Approximation for Subadditive Valuations}
\label{sec:subadd-theorem}

Thus far we have shown through a Revenue Recovery Scheme construction that there exists a sequential item pricing mechanism that achieves an $\bigoh(\log m + \log p_{max}/p_{min})$ approximation with respect to $\earev(\D,\itempricing)$, whenever the valuations in $\D$ are subadditive. Naturally, this dependency on the aspect ratio of the ex ante prices is not desirable, as it can potentially become very large. In this section, we show how to remove this dependency, and prove~\Cref{thm:maintheorem}, which we re-state for convenience:

\maintheorem*
The key observation resulting in~\Cref{thm:maintheorem} is that the $\log(p_{max}/p_{min})$ factor we suffer in the design of an RRS is manifested through the set of items we attempt to sell to the buyer; in other words, if we only sell items whose ex ante optimal prices lie in some window of aspect ratio $\mathrm{poly}(m)$, then our approximation guarantee would indeed be $\bigoh(\log m)$. The problem is that by restricting what items we are allowed to sell, we will lose the potential revenue that could be collected from them. Unsurprisingly, the revenue lost from not selling \textit{very low price items} can be neglected, as it consists only a small portion of the total revenue that we wish to approximate. However, the same argument cannot be made for the revenue lost from not selling \textit{very high price items}, as they can potentially contribute an arbitrary large fraction of the total revenue. In fact, in Appendix~\ref{sec:rrs-lb} we show that this loss is necessary for some XOS value distributions -- no RRS can obtain more than 
a $\operatorname{poly}(m)$ fraction of the desired revenue. We circumvent this obstacle by designing a different sequential item pricing mechanism that captures the revenue contributed by the very high price items by exploiting the fact that these items are sold with polynomially small probabilities.

We will now make the above approach concrete. Consider any joint distribution $\D$ and for convenience, let $\obj =\earev(\D,\itempricing)$ denote the objective we wish to approximate. Let $p=(p_1,\dotsc , p_n)$ be the vector of (random) item pricings that achieve the optimal ex ante revenue $\obj$.
We will partition prices into ``small", ``medium", and ``large" sets as follows. Define
\[S \defeq \left[0, \frac{\obj}{m^2}\right),\;\;\;\;
M \defeq \left[\frac{\obj}{m^2}, 8m^2\obj\right],\;\;\;\; L \defeq \left(8m^2\obj,+\infty\right).\]
Note that $\obj = \objs + \objm + \objl$, where
\[\objs \defeq \expect_{v \sim \D, p}\left[ \sum_{i=1}^n\sum_{j=1}^m p_{ij} \alloc_j(v_i,p_i)\cdot\mathbbm{1}(p_{ij}\in S)\right],\]
is the fraction of $\obj$ recovered via items with optimal ex ante prices in $S$, and $\objm$ and $\objl$ are defined analogously for $M$ and $L$. We first note that $\objs$ is a negligible fraction of $\obj$. Indeed,
\begin{align*}
    \objs \leq \frac{\obj}{m^2}\cdot \expect_{v \sim \D, p}\bigg[\sum_{i,j}\alloc_j(v_i,p_i)  \bigg]  = \frac{\obj}{m^2}\cdot\sum_{j=1}^m\sum_{i=1}^n\alloc_j(\D_i, p_i)
    \leq \frac{\obj}{m^2}\cdot\sum_{j=1}^m 1
    \leq \frac{\obj}{m},
\end{align*}
with the second-to-last inequality following from the ex-ante condition. From this, observe that if we can obtain an $\bigoh(\log m)$-approximation to $\objm$ and $\objl$ via two separate sequential item pricing mechanisms, then by flipping a fair coin at the beginning to decide which mechanism to run, we will obtain a new sequential item pricing mechanism whose expected revenue also approximates $\obj$ up to a factor of $\bigoh(\log m)$; this is precisely the mechanism that achieves the guarantee of~\Cref{thm:maintheorem}.

As already mentioned, handling $\objm$ is easy using the machinery we established in the previous sections. When buyer $i$ arrives, instead of attempting to sell any available item in $S_i$, we will only attempt to sell items $j\in S_i\cap M(p_i)$ where $p_i\in\Rplus^m$ is the (sampled) ex ante price for the buyer and $M(p_i)$ denotes the set of ``medium" priced items, i.e. $M(p_i)\defeq \{j: p_{ij}\in M\}$. We can then apply the $\bigoh(\log m + \log p_{max}/p_{min})$-OCRS guaranteed by the combination of \Cref{t:rrs_argument} and \Cref{thm:subadd-RRS} to $S_i\cap M(p_i)$. Using the fact that the price ratio between any two prices in $M$ is at most $8m^4$, we will obtain a (random) pricing $q_i$ such that $\alloc(\D_{i|S_i\cap M(p_i)}, q_i) \preceq \alloc (\D_i,p_i)$ and $\sum_{j\in S_i} p_{ij} \alloc_j(\D_i,p_i)\leq \bigoh(\log m)\cdot \rev(\D_{i|S_i\cap M(p_i)}, q_i)$. By taking the corresponding expectation over the sampled $p_i$, we obtain a ${\bigoh(\log m)}$ approximation to $\objm$, as desired. 

Finally, we are left with the task of obtaining a $\bigoh(\log m)$ approximation to $\objl$; in fact, we will show that a constant approximation is possible for all distributions over monotone valuation functions.

\begin{lemma}\label{lem:resolve-obj-l}
For any value distribution $\D$ and any item pricing $p$, there exists an item pricing $q$ such that $q_j \ge 2m \cdot \obj$ for all $j\in [m]$, and
\[\rev(\D,q)\ge \frac{1}{4}\cdot \sum_{j=1}^m p_j \alloc_j(\D, p)\cdot \mathbbm{1}[p_j\in L].\]
\end{lemma}

Before proving~\Cref{lem:resolve-obj-l}, let us first argue why it directly implies a $\bigoh(1)$ approximation to $\objl$. Consider the sequential item pricing mechanism that upon buyer $i$'s arrival observes the set of available (i.e. unsold) items $S_i$ and proceeds as follows. If at least one item is unavailable, it skips the buyer, selling him nothing. Otherwise, it samples a pricing $p_i$ from the optimal ex ante pricing for buyer $i$, and presents the buyer with the pricing $q_i$ guaranteed by \Cref{lem:resolve-obj-l}. We will prove that with constant probability this process sells no items, and therefore, obtains in expectation a constant fraction of the revenue contribution of each buyer $i$ to $\objl$. Formally:

\begin{claim}\label{c:obj-l-maintain-prob}
    For each buyer $i$, $S_i=[m]$ with probability at least $1/2$.
\end{claim}

It remains to prove~\Cref{lem:resolve-obj-l} and~\Cref{c:obj-l-maintain-prob}.

\medskip
\begin{proof}[Proof of~\Cref{lem:resolve-obj-l}] 
We define the item pricing $q$ as follows.
\[
    q_j \defeq \begin{cases}
        p_j/2 & \text{if $p_j \in L$, i.e. $p_j \ge 8m^2 \cdot \obj$}, \\
        \max\{p_j, 2m \cdot \obj\} & \text{otherwise}.
    \end{cases}
\]
It is straightforward to see that $q_j \ge 2m \cdot \obj$ for every $j$. To prove the lemma, it suffices to show that for any fixed buyer $v$ it holds that
\[\rev(v,q) \geq \frac{1}{4}\cdot \sum_{j=1}^mp_j\alloc_j(v,p)\cdot \mathbbm{1}(p_j\in L),\]
and then simply taking an expectation over $v\sim\D$ will produce the claim. Fix the buyer $v$ and define $T \defeq \alloc(v,p)$ and $T' \defeq \alloc(v,q)$ to be utility-maximizing sets under pricings $p$ and $q$ respectively. We will overload notation and use $L$ to denote the set of indices $j$ where $p_j\in L$. Then note that the LHS of the above claimed inequality is $q(T')$ and the sum in the RHS is $p(T\cap L)$. So, we need to show that 
\[q(T') \geq \frac{1}{4} \cdot p(T \cap L).\]
This clearly holds when $T\cap L=\emptyset$, so we assume that $T\cap L\neq \emptyset$. On one hand,
 \begin{align}\label{eq:lem-15-first}
        p(T) - q(T) = p(T \cap L) - q(T \cap L) + p(T \setminus L) - q(T \setminus L)
        \geq \frac{1}{2} \cdot p(T\cap L) - 2m^2 \cdot \obj 
        \geq \frac{1}{4}\cdot p(T \cap L)
\end{align}
where the first equality uses additivity of prices; the middle inequality uses that $p_j = 2q_j$ for all $j \in L$, while $q_j \ge 2m \cdot \obj$ for all $j \notin L$; and the last inequality uses that $p(T \cap L) \ge 8m^2 \cdot \obj$ as $T \cap L$ is nonempty. On the other hand,
\begin{align}\label{eq:lem-15-second}
    p(T') - q(T') = p(T' \cap L) - q(T' \cap L) + p(T' \setminus L) - q(T' \setminus L)
    \leq q(T' \cap L)
    \leq q(T')
\end{align}
where we used $q_j \ge p_j$ for $j \notin L$. 
Finally, as $v(T) - p(T) \ge v(T') - p(T')$ and $v(T') - q(T') \ge v(T) - q(T)$ by the definition of $T$ and $T'$, we have $q(T) - q(T') \ge v(T) - v(T') \ge p(T) - p(T')$, or
\begin{align}\label{eq:lem-15-third}
    p(T') - q(T') \ge p(T) - q(T).
\end{align}
The lemma then follows from combining~\Cref{eq:lem-15-first,eq:lem-15-second,eq:lem-15-third}.
\end{proof}

\medskip
\begin{proof}[Proof of~\Cref{c:obj-l-maintain-prob}]
It suffices to show that $\sum_{i'<i} \sum_{j=1}^m \alloc_j(\D_{i'}, q_{i'}) \le \frac{1}{2}$, from which our claim follows from a union bound. Notice that here $q_{i'}$ is a random variable that depends on the pricing $p_{i'}$ that our algorithm draws.

We first note that $\alloc_j(\D_{i'}, q_{i'}) \le \frac{1}{2m}$ for every buyer $i'$\footnote{Observe that from this, if the ex ante allocations are at least $1/(2m)$ then we can cast this algorithm as a $O(1)$-OCRS w.r.t. $\objl$. Combined with the $O(\log m)$-OCRS w.r.t. $\objm$, this justifies the claim of Informal Theorem~\ref{it:ocrs}.}; otherwise, as $q_{i', j} \ge 2m \cdot \obj$, by simply using $q$ for the buyer $i'$ with valuation $\D_{i'}$ (and not selling anything to any other buyer) we would get an expected revenue that is strictly larger than $\obj\defeq \earev(\D,\itempricing)$ which is obviously a contradiction by the optimality of the ex-ante solution. This means that for every buyer $i'$, we have
\[\sum_{j=1}^m \alloc_j(\D_{i'}, q_{i'}) \le \frac{1}{2}.\]

We now prove the claim via contradiction. Suppose the claim is not true for some $i\in [n]$. Then there must exists a first index $k \leq i$ where it becomes not true, namely, \[\sum_{i'<k} \sum_{j=1}^m \alloc_j(\D_{i'}, q_{i'}) \le \frac{1}{2}\;\;\;\text{ and }\;\;\;\sum_{i'\le k} \sum_{j=1}^m \alloc_j(\D_{i'}, q_{i'}) > \frac{1}{2}.\] 
Note that as $\sum_{j=1}^m \alloc_j(\D_{k}, q_{k}) \le \frac{1}{2}$, we have \[\sum_{i'\le k} \sum_{j=1}^m \alloc_j(\D_{i'}, q_{i'}) \le 1.\]

Therefore, the pricing vector $q'$ which is equal to $q$ on the first $k$ buyers and to infinity on the remaining buyers satisfies the ex-ante constraint, and its ex-ante revenue will be at least 

\[\expect\left[\sum_{i' \le k} \sum_{j=1}^m q_j \cdot \alloc_j(\D_{i'}, q_{i'})\right]\ge \frac 12\cdot 2m\cdot \obj \ge m \obj,\] 
which contradicts ex-ante optimality.
\end{proof}

%% file: sec_computational_aspects.tex
\section{Computational Considerations}\label{sec:computational-aspects}
In this section, we summarize Sections~\ref{sec:rrs-subadd} through~\ref{sec:subadd-theorem} to present a top-down overview for the existence of the sequential item pricing mechanism that achieves the guarantees of~\Cref{thm:maintheorem}, namely the $O(\log m)$-approximation to the ex ante revenue under subadditive valuations, as well as comment on the computational aspects of our approach.

\paragraph{Computational Aspects.} Recall that our objective is to approximate
\[
    \earev(\D,\itempricing) = \sum_{i=1}^n\earev_{x^*_i}(\D_i,\itempricing)
\]
under some (optimal) ex-ante allocation constraint $\{x^*_i\}_{i=1}^n$. Our approach crucially relies on access to this optimal solution, as well as the distributions (over item pricings) $p^*_i$ that achieve $\earev_{x^*_i}(\D_i,\itempricing)$ for each buyer $i\in [n]$. We comment that finding the revenue-optimal item pricing is hard even in the absence of an ex ante constraint \cite{chalermsook2013independent}, but approximations may exist in special cases (e.g., \cite{10.1145/3519935.3520065}). 

Our approach also requires access to an oracle that, given distribution $\D_i$ and pricing $q$, can compute the expected allocation $\alloc(\D_i,q)$ and expected revenue $\rev(\D_i,q)$. Note that if the supports of the buyers' value distributions are small and we have access to a demand oracle over subadditive values, we can compute these quantities. 

We remark that assuming access to $x^*_i$, $p^*_i$, and the demand oracle as discussed above, the remainder of our algorithm is \textit{fully constructive}, i.e. we can efficiently construct the Sequential Item Pricing achieving our claimed bounds in time polynomial in $n$, $m$, and the support sizes of the value distributions.

\paragraph{Construction Overview.} As described in Section~\ref{sec:subadd-theorem}, our construction begins by flipping a fair coin in order to decide whether it will focus on \textit{high-price items} or \textit{medium-price items} throughout the execution. 

If the coin lands on high-price items, we proceed to address the buyers sequentially, until we sell at least one item to a buyer; once this happens, we completely ignore the rest of the buyers. To address buyer $i$, we sample an item pricing $p_i$ from distribution $p^*_i$ and present them with the pricing $q_i = q_i(p_i)$ that we define in~\Cref{lem:resolve-obj-l}. Assuming access to the value of the (optimal) ex ante revenue, this construction is straightforward and thus the entire sequential item pricing is constructed in $O(nm)$ time. In fact, since $q_i$ doesn't depend on the set of available items, we can pre-compute these prices in advance, i.e. our sequential item pricing is oblivious to the order in which the buyers arrive, as well as the randomness of their valuations.

If the coin lands on medium-price items, then whenever buyer $i$ arrives we observe the set of available items $S_i$ and offer them a realization of the (random) item pricing $q_i=\mathrm{OCRS}(\D_i, S_i, \frac{1}{2}x^*_i)$. Note that this sequential item pricing is adaptive both to the order in which the buyers arrive, and to the realizations of their values, as both affect the set of available items $S_i$. Furthermore, to avoid having to sample from the optimal pricing distribution that respects $\frac{1}{2}x^*_i$, we can instead present them with the pricing $q_i=\mathrm{OCRS}(\D_i, S_i,x^*_i)$ with $1/2$ probability (and entirely skip the buyer with the remaining $1/2$ probability); it is easy to see that the proof of Section~\ref{sec:OCRStoEARev} also holds for this alternative.

In order to construct $q_i=\mathrm{OCRS}(\D_i, S_i,x^*_i)$, we sample an item pricing $p_i$ from distribution $p^*_i$. We then compute the vector $w_j=p_{ij}\alloc_j(\D_i,p_i)$ for $j\in S_i$ as in Section~\ref{sec:rrs_to_ocrs}, and then run the algorithm $\chalg$. An important observation is that instead of pre-computing the $y^T$ vectors for all $T\subseteq S_i$, we can compute them on demand; then, since $\chalg$ clearly runs for at most $m$ steps, we only need to compute $m$ of these vectors.

To compute a $y^T$ vector, we need to find a pricing $q^T$ satisfying the revenue guarantee, and then compute its revenue. \
From Section~\ref{sec:rrs-subadd}, we know that picking the best scaling $\gamma\in [\ell,h]$ and setting $q=\gamma\cdot p$ suffices. In fact, because of the linear dependence of the price of any particular set of items on $\gamma$, at the loss of a factor of $2$ it suffices to choose $\gamma$ as the best power of $2$ in the range $[\ell,h]$. 

Once $\chalg$ completes its execution and returns the distribution $\lambda_T$ over the subsets $T\subseteq S_i$, we simply sample such a set $T$ from it and return $q_i=q^T$.

%% file: sec_xos_lb.tex
\section{Lower Bounds for XOS Valuations (Proof of \Cref{thm:xos-lb})}\label{sec:xos-lb}

In this section, we prove that the gap between sequential item pricings and ex ante item pricings can be as large as $\Omega(\sqrt{\log m})$ under XOS valuations, as stated in Theorem~\ref{thm:xos-lb}, which we restate for convenience:

\xoslbthm*

In fact, a stronger separation holds. Let $\expostitempricing$ be the class of multi-buyer mechanisms which take all buyers' valuations $v_i$ as input, and offer each buyer a randomized item pricing $\mech_{i, v_{-i}}$ over the items which depends only on $\D_i$ and the other buyers' valuations $v_{-i} \defeq \{v_{i'}: i' \neq i\}$. Provided that the allocations under the $\mech_{i, v_{-i}}$ are feasible, these are feasible and truthful multi-buyer mechanisms. Furthermore, this generalizes $\seqitempricing$, since any randomized sequential item pricing mechanism can be simulated by a mechanism in $\expostitempricing$. We will show that the separation holds for this class of mechanisms as well.
\begin{theorem}\label{t:xos-lb-offline}
    There exists a joint value distribution $\D$ for $m$ items and $n=\sqrt{m}$ buyers for which
    \[
        \earev(\D,\itempricing) \geq \Omega\left(\sqrt{\log m}\right) \cdot \rev(\D, \expostitempricing).
    \]
\end{theorem}
\noindent \Cref{thm:xos-lb} then follows as a corollary of \Cref{t:xos-lb-offline}, since every mechanism $\mech \in \seqitempricing$ may be simulated by a mechanism in $\expostitempricing$.

The rest of this section is devoted to proving~\Cref{t:xos-lb-offline}. We begin by specifying our instance. We first define a set of parameters that will be useful towards stating our claims; let $k\defeq\sqrt{m}$, and $t\defeq \sqrt{\log k}$, taking logarithms to be base two. By incurring at most a constant factor, we can assume that all of these parameters are integer-valued, and also that $t$ is even. 

\paragraph{The buyers.} Conveniently, the $n$ buyers $v_i \sim \D_i$ are all (independently) identically distributed, with valuations sampled according to the following process:
\begin{enumerate}
    \item A set of items $A\subseteq [m]$ with $|A|=k$ is sampled uniformly at random.
    \item An integer $h \in \{1,2,\dotsc , \frac{1}{2} \log k\}$ is drawn uniformly at random. We denote $\ell\defeq2^h$.
    \item The valuation of the buyer is then realized as an XOS function, which is given by the maximum over the following additive valuations:
    \begin{enumerate}
        \item For the set $A$, we define an additive valuation $v_A$ such that $v_A(j)= 1 + t$ for items $j\in A$ and $v_A(j)=0$ for items $j\notin A$.
        \item For \textit{every} set $B$ of cardinality $|B| = t\cdot \ell$, we define an additive valuation $v_B$ such that $v_B(j)=1 + \frac{k}{\ell}$ for items $j\in B$ and $v_B(j)=0$ for items $j\notin B$. 
    \end{enumerate}
\end{enumerate}
Observe that for our selected parameters $t\cdot \ell \leq m$ always, and all parameters are integer-valued, so this is a well-defined XOS instance. Let $A_i$ be the $A$-set for buyer $i$. 

Informally, this valuation ensures that each buyer $i$ is nearly indifferent between the (single) high-revenue set $A_i$ on the one hand, and one of the (many) low-revenue sets $B$ on the other. For any sequential item pricing mechanism, as the set of available items shrinks, $A_i$ becomes less valuable to $i$, while the most valuable low-revenue set $B$ does not. As a result, it quickly becomes unlikely that the available fraction of $A_i$ is large enough for $i$ to choose it, resulting in reduced revenue.

\paragraph{Ex-ante revenue.} To demonstrate a gap, we need a lower bound on $\earev(\D,\itempricing)$ for the above instance $\D$. The key observation is that ex ante, all buyers can be allocated their $A$-sets. Consider the item pricing $p_{ij}=1$ for all buyers $i\in [n]$ and items $j\in [m]$, meaning that each buyer is presented with the exact same set of uniform prices over the items. For this pricing, the utility of the $A$ set and any $B$ sets are precisely
\begin{align*}
    \util_i(A) &= v_i(A) - p_i(A) = k(1+t-1) = kt, \;\;\text{ and,} \\
    \util_i(B) &= v_i(B) - p_i(B) = t\ell\left(1 + \frac{k}{\ell}-1\right) = kt.
\end{align*}
By breaking ties for $A$ (or by slightly increasing $A$'s valuation by $\epsilon > 0$) we can assume the buyer will always buy set $A$. Since $A$ is a uniformly random subset of $[m]$ with cardinality $k$, each item will therefore be sold to the buyer with probability $k/m = m^{-1/2}$. Since we have $n=\sqrt{m}$ identical buyers in total, this means that the fractional total allocation of each item is precisely $1$, and thus the ex ante constraint is satisfied for our proposed pricing $p$, implying that
\[
    \earev(\D,\itempricing) \geq \sum_{i=1}^n\rev (\D_i, p_i).
\]
Since we have already showed that under this $p_i$ buyer $i$ always purchases $A_i$, the per-buyer revenue is simply $\rev(\D_i,p_i) = |A_i|=k$ and so for our choice of parameters $n=k=\sqrt{m}$,
\begin{equation} \label{eq:xos-lb-EA-rev}
    \earev(\D,\itempricing) \geq n\cdot k = m.
\end{equation}

\paragraph{Revenue and allocation for single-buyer item pricings.}

The following are crucial to analyzing this instance. They establish conditions under which a buyer can be sold items according to their $A$-type versus $B$-type additive valuations, and upper bounds on the expected revenue in each case. Omitted proofs appear in \Cref{sec:xos-lb-proofs}. 
We begin with claims that hold for all valuations $v\sim \D_i$ from the above distribution, all item pricings $p$, and all subsets of available items $S \subseteq [m]$. Let $\E^A(v_i,p)$ denote the event that buyer $v_i$ facing prices $p$ chooses an allocation where their valuation is supported by $v_A$; define $\E^B(v_i,p)$ analogously when their valuation is supported by any of the $v_B$.

\begin{restatable}{lemma}{xosfirstlemma}\label{c:xos-lb-A-items}
    For any $v_i,p$ such that $\E^A(v_i,p)$ holds, buyer $i$ purchases at least $\frac{kt}{t+1}$ items.
\end{restatable}

\begin{restatable}{lemma}{xossecondlemma}\label{c:xos-lb-A-rev}
    For some $c_1\in \Rplus$ and sufficiently large $k$, for all $p$, $\xpectover{v_i}{\rev(v_i, p) \mid \E^A(v_i, p)} \leq c_1 \cdot k$.
\end{restatable}

\begin{restatable}{lemma}{xosthirdlemma}\label{c:xos-lb-B-rev}
    For some $c_2\in \Rplus$ and sufficiently large $k$, for all $p$, $\xpectover{v_i}{\rev(v_i, p) \mid \E^B(v_i,p)} \leq c_2\cdot t\cdot\frac{k}{\log k}$.
\end{restatable}
\noindent These will enable us to reason about item-pricing mechanisms in multi-buyer settings for this instance.

\paragraph{Revenue and allocation for multi-buyer item pricings.}

\Cref{c:xos-lb-A-items} implies that no buyer facing an item pricing will choose an $A$-type allocation unless they receive a large proportion of their $A$-set $A_i$. We will now argue that with high probability only a few buyers can receive large proportions of their $A$-sets (and therefore choose a $A$-type allocations), even if the mechanism is given access to $A_1, \ldots, A_n$ upfront.

We will make use of the abstraction of an \emph{assignment}, which is an arbitrary mapping of items to buyers.
\begin{definition} \label{def:item-assn}
    Let $\sigma: [m] \rightarrow [n] \cup \{\bot\}$ be an \emph{assignment} of items to buyers.
\end{definition}
Note that every multi-buyer mechanism---item-pricing or otherwise---induces a distribution over assignments, where $\sigma^{-1}(\bot)$ are the items not allocated to any buyer. We will call an allocation \emph{feasible} for a buyer if it allocates them at least a $(1-c)$-proportion of their $A_i$:
\begin{definition} \label{def:feasible-assn}
    A buyer $i \in [n]$ is \emph{$(1-c)$-feasible under} $\sigma$ if $|A_i \cap \sigma^{-1}(i)| \geq (1-c)n$. For $I \subseteq [n]$, $I$ is \emph{$(1-c)$-feasible under} $\sigma$ if all $i\in I$ are $(1-c)$-feasible under $\sigma$.
\end{definition}

Using these definitions we present the following combinatorial lemma, which together with \Cref{c:xos-lb-A-items} implies that for any multi-buyer mechanism which presents buyers with item pricings, the number of buyers with $A$-type allocations is small with high probability:

\begin{restatable}{lemma}{xosfourthlemma}\label{lem:max-t-feasible-num-buyers}
    Let $I_{max}(v)$ be the largest $I \subseteq [n]$ for which $I$ is $(1-\frac{1}{t+1})$-feasible under $\sigma$, for \emph{any} assignment $\sigma$. 
    Then for some $c_3$ and sufficiently large $k$, $\xpectover{v}{|I_{max}(v)|} \leq c_3 \cdot \frac{n}{t}$.
\end{restatable}

We are now ready to prove \Cref{t:xos-lb-offline}.
\begin{proof}[Proof of \Cref{t:xos-lb-offline}]
    The proof is intuitive: with high probability, no mechanism can allocate $A$-sets to more than a $(1/t)$-proportion of the buyers and get revenue $k$ from each, and since it presents buyers with item pricings, its revenue on the remaining buyers is at most $tk/\log k$ each. Then, $t=\sqrt{\log k}$ balances these terms, resulting in the gap.
    
    Formally, fix an ex post randomized item pricing mechanism $\mech  \in \expostitempricing$. Then $\mech$ is a collection of randomized item pricings $\mech_{i,v_{-i}}$, and the expected revenue is 
    \begin{align*}
        \rev(\D,\mech ) &= \xpectover{v}{ \sum_i \left( \xpectover{p_i\sim \mech_{i,v_{-i}}}{\rev(v_i, p_i)} \right) } \\
        &= \xpectover{v}{ \sum_i \left( \xpectover{p_i\sim \mech_{i,v_{-i}}}{\xpectover{v_i}{\rev(v_i, p_i)}} \right)}, 
        \intertext{because $\mech_{i, v_{-i}}$ is independent of $v_i$. Then conditioning on the events $\E^A \defeq \E^A(v_i,p_i)$ and $\E^B\defeq \E^B(v_i,p_i)$,}
        &= \xpectover{v}{ \sum_i \left( \xpectover{p_i\sim \mech_{i,v_{-i}}}{\xpectover{v_i}{\rev(v_i, p_i) \mid \E^A} \cdot \probover{v_i}{\E^A} +  \xpectover{v_i}{\rev(v_i, p_i) \mid \E^B} \cdot \probover{v_i}{\E^B} } \right) } \\
        &\leq \xpectover{v}{ \sum_i \left( \xpectover{p_i\sim \mech_{i,v_{-i}}}{c_1 \cdot k \cdot \probover{v_i}{\E^A(v_i,p_i)} + c_2 \cdot \frac{tk}{\log k} } \right)} \\
        &= c_2 \cdot \frac{ntk}{\log k} + c_1 \cdot k \cdot \xpectover{v}{ \sum_i \left( \xpectover{p_i\sim \mech_{i,v_{-i}}}{ \probover{v_i}{\E^A(v_i,p_i)} } \right)},
        \intertext{for some constants $c_1, c_2$, by \Cref{c:xos-lb-A-rev,c:xos-lb-B-rev}. Next, observe that for fixed $v,p$ the allocations $\alloc(v_i, p_i)$ form an assignment $\sigma$, since $\mech$ is a feasible mechanism. Let $I(v, p)$ be the set of buyers for which this assignment is $(1-\frac{1}{t+1})$-feasible. Then \Cref{c:xos-lb-A-items} implies that $\E^A(v_i,p_i) \subseteq \{i \in I(v,p)\}$, so}
        &\leq c_2 \cdot \frac{ntk}{\log k} + c_1 \cdot k \cdot \xpectover{v}{ \sum_i \left( \xpectover{p_i\sim \mech_{i,v_{-i}}}{ \probover{v_i}{i \in I(v,p)} } \right)} \\
        &= c_2 \cdot \frac{ntk}{\log k} + c_1 \cdot k \cdot \xpectover{v}{ \xpectover{p\sim \mech}{ |I(v,p)|}}.
    \end{align*}
    Next let $I_{max}(v)$ be a maximum-size $(1-\frac{1}{t+1})$-feasible subset of the buyers $v$ over \emph{any} assignment $\sigma: [m] \rightarrow [n]\cup \{\bot\}$, and note that $|I(v,p)| \leq |I_{max}(v)|$. Then by \Cref{lem:max-t-feasible-num-buyers} we have that $\xpectover{v}{|I_{max}|} \leq c_3 n/t$ for some constant $c_3$. Therefore
    \begin{align*}
        \rev(\D,\mech) &\leq c_2 \cdot \frac{ntk}{\log k} + c_1 \cdot k \cdot \xpectover{v}{|I_{max}(v)|} \\
        &\leq c_2 \cdot \frac{ntk}{\log k} + c_1 \cdot k \cdot \frac{c_3 \cdot n}{t}.
    \end{align*}
    By our choice of parameters $n=k=\sqrt{m}$ and $t = \sqrt{\log k}$ this implies that $\rev(\D,\mech ) = O(m \log^{-1/2}m) $, while by \eqref{eq:xos-lb-EA-rev} the ex ante item pricing revenue for $\D$ is at least $m$. This holds for all $\mech \in \expostitempricing$, completing the proof.
\end{proof}

%% file: sec_related.tex
\section{Related Work} 
\label{sec:related}

\paragraph{Multi-parameter Revenue Maximization.} As mentioned previously, multi-parameter revenue maximization is a notoriously hard problem. \citet{BCKW-JET15} and \citet{hart2013menu} showed that the optimal revenue is inapproximable within any finite factor by simple mechanisms (such as item pricings) even with just two items and one unit-demand or additive buyer. Much of the literature on simplicity versus optimality for revenue maximization accordingly makes strong assumptions on the buyers' value distributions, such as requiring the values to be independent across items. Under these assumptions, mechanisms such as sequential item pricing,  grand bundle pricing and two part tariffs obtain constant factor approximations to the optimal revenue. See, e.g.,
\cite{chawla2015power,babaioff2014simple,rubinstein2018simple,10.1145/3055399.3055465, cai2019duality}.

There are several connections between this literature and our work. \citet{chawla2010multi} showed that under the item independence assumption, the revenue maximization problem for {\em unit demand} buyers under various feasibility constraints can be reduced to a single-parameter prophet inequality. They were the first to use an ex ante relaxation in multi-buyer settings, a technique that was later formalized by \citet{doi:10.1137/120878422}. \citet{chawla2016mechanism} extend a similar approach based on the ex ante relaxation and prophet inequalities to matroid rank values, still with an independence-across-items assumption.

It is worthwhile contrasting our main results with those of \citet{10.1145/3055399.3055465}. They consider buyers with subadditive values and bound the gap between the revenue of a sequential two-part tariff mechanism and the ex ante relaxation. There are several similarities and differences from our work: (1) Like the papers referenced above, Cai and Zhao require values to be independent across items, a significant restriction. Our results do not require any distributional assumptions. (2) Under the independence assumption however, they are able to compete against the ex ante optimal revenue, where our approximations only apply to the ex ante buy-many revenue, a weaker benchmark. (3) Cai and Zhao's mechanism is a two-part tariff, meaning that every buyer needs to pay an {\em entry fee} to enter the mechanism and is then offered an item pricing. Two-part tariffs are necessary to obtain constant factor approximations to the optimal revenue. (4) Finally, Cai and Zhao obtain a constant factor gap for XOS valuations and an $O(\log m)$ gap for subadditive valuations, whereas we obtain an $O(\log^2 m)$ gap for both. 
The $O(\log m)$ factor for subadditive values was subsequently improved to an $O(\log\log m)$ factor by \citet{doi:10.1137/20M1382799} via the construction of a new prophet inequality for welfare maximization over subadditive values.

\paragraph{Buy-Many Mechanisms.} Another line of work on simple versus optimal multi-parameter mechanisms considers approximating a weaker benchmark, namely \emph{buy many} mechanisms, rather than introducing assumptions on the value distributions. This line of work was introduced by \citet{BCKW-JET15} for unit-demand buyers and extended to general buyers by \citet{chawla2022buy}.  These papers consider the single-buyer mechanism design problem and bound the gap between the revenue of item pricing and the optimal buy many mechanism. This gap is shown to match the best possible gap achievable by {\em any} mechanism with subexponential description complexity. \citet{chawla23buy} were the first to consider buy-many mechanisms in multi-buyer settings. They argue that any reasonable extensions of buy-many mechanisms to multiple buyers are upper bounded by ex ante buy-many revenue. As mentioned previously, this work bounds the gap between sequential item pricing revenue and the ex ante buy-many revenue, a result we extend to subadditive buyers.

\paragraph{Prophet inequalities and contention resolution.}
Prophet inequalities were first introduced to mechanism design by \citet{10.5555/1619645.1619656}. \citet{chawla2010multi} showed how to apply them to revenue maximization settings and motivated the design of prophet inequalities for more general settings. Further work on prophet inequalities has considered variants with different feasibility constraints, arrival models, limited information settings, buyer correlations, etc. We refer the reader to the survey by \citet{10.1145/3144722.3144725} for an overview of this work.

As mentioned previously, single-parameter prophet inequalities have strong connections with Online Contention Resolution Schemes. Contention Resolution Schemes were first defined by \citet{vondrak11submodular} in the context of offline maximization of set functions under various feasibility constraints as an abstraction for rounding the fractional relaxations of these objectives. Subsequently \textcite{feldman16online} defined {\em Online} CRSes as analogous online rounding primitives for a range of online adversaries. \citeauthor{feldman16online} observed that good OCRSes imply good prophet inequalities for combinatorial constraints and linear objectives. Using duality, \textcite{lee18optimal} showed that the converse also holds -- good prophet inequalities imply good OCRSes.

\paragraph{Combinatorial prophet inequalities.}
Prophet inequalities have also been studied extensively in the multi-item setting for the social welfare objective where buyers have combinatorial valuations over the items. This line of work was introduced by \citet{doi:10.1137/1.9781611973730.10}. Inspired by the matroid prophet inequality of \citet{10.1145/2213977.2213991}, \citeauthor{doi:10.1137/1.9781611973730.10} introduced the approach of using {\em balanced prices} to design prophet inequalities. This has become the dominant technique for designing combinatorial prophet inequalities (see, e.g., \cite{doi:10.1137/20M1323850, doi:10.1137/20M1382799, doi:10.1137/20M1323850}), although some constructions use duals or other approaches for constructing  prices \cite{ZHANG2022143, 10.1145/3564246.3585151}. We remark that all of these approaches are quite different from OCRSes in that they do not necessarily preserve the probabilities of allocation of the ex ante relaxation. Finally, we note that for the social welfare objective over subadditive values, \citet{10.1145/3564246.3585151} recently developed a constant factor prophet inequality. 

%% file: sec_conclusion.tex
\section{Conclusion and Open Problems}\label{sec:conclusion}

Our work develops polylogarithmic upper and lower bounds on the ratio between the revenue obtained by sequential item pricings and the ex ante item pricing (or buy many) revenue. Although the $O(\log m)$ gap between the sequential item pricing revenue and the ex ante buy-many revenue is tight for unit-demand valuations, it is not clear whether an additional $\log m$ factor is necessary to lose for XOS or subadditive values. 

Another important direction is to make our result algorithmic. Constructing our item pricing requires access to the optimal solution to the ex ante item pricing relaxation. While optimizing item pricing revenue is hard even in the absence of ex ante constraints, {\em approximations} may be possible under suitable assumptions. Another possible avenue is to first construct an approximately optimal ex ante buy-many mechanism and then convert it into an approximately optimal ex ante item pricing. The complexity of optimizing revenue over buy-many mechanisms is not well understood. Furthermore, we note that the upper bound on the gap between the ex ante item pricing and buy many revenues (due to \cite{chawla23buy}) is non constructive. 

%% file: sec_prelims_MD.tex
\section{Mechanism Design Basics}
\label{sec:PrelimMD}

In this section we cover the necessary background on mechanism design that is required for this paper. We consider a setting with $m$ items and $n$ buyers with combinatorial valuations over the items, with the objective of maximizing the total revenue of the seller. We begin by establishing some notation.

\paragraph{Buyers and valuations.} A \textit{buyer} is modeled through a \textit{valuation function} $v:2^{[m]}\rightarrow \Rplus$ that is drawn from a known distribution $\D$ and assigns a non-negative value to each subset of the $m$ items. Some of our results reference special classes of valuations that we define below.
\begin{itemize}
    \item Unit-Demand: A valuation $v$ is \emph{unit-demand} if 
    $v(S) = \max_{j \in S} v(\{j\})$ for all $S \subseteq [m]$. For convenience, we write $v(\{j\})$ as simply $v_j$.
    \item Gross Substitutes: Let $p$ and $p'$ be two item pricings with $p\preceq p'$ and let $S=\{j: p_j=p'_j\}$ denote the set of items where the two pricings charge equal prices. We say that a value function $v$ satisfies \emph{gross substitutes} if for all such pricings $p$ and $p'$, $\alloc(v,p)\cap S\subseteq \alloc(v,p')\cap S$. In other words, under the new pricing $p'$ the buyer continues to purchase those items in $\alloc(v,p)$ whose prices did not increase.
    \item Fractionally Subadditive (XOS): A valuation $v$ is \emph{fractionally subadditive} if there exists a set $\mathcal{A}$ of $m$-dimensional vectors $a \in \Rplus^m$, such that $v(S) = \max_{a \in \mathcal{A}} \sum_{j \in S} a_j$ for all $S\subseteq [m]$.  
    \item Subadditive: A valuation $v$ is \emph{subadditive} if $v(S) + v(T) \geq v(S \cup T)$ for all $S, T \subseteq [m]$.  
    \item Monotone: A valuation $v$ is \emph{monotone} if $v(S) \leq v(T)$ for all $S\subseteq T \subseteq [m]$.
\end{itemize}
Observe that we have $\text{Unit-Demand}\subsetneq \text{Gross Substitutes} \subsetneq \text{XOS}\subsetneq \text{Subadditive}\subsetneq \text{Monotone}$.

Finally, for any valuation function $v$ and a subset  $S\subseteq [m]$ of items, we denote by $v_{|S}$ the valuation function given by $v_{|S}(T)=v(T\cap S)$ for all $T\subseteq [m]$;
analogously, we denote by $\D_{|S}$ the value distribution that first samples a valuation $v\sim \D$, and then returns $v_{|S}$. In other words, $\D_{|S}$ captures the valuation of buyer $\D$ if we can only offer them items in $S$.

\paragraph{Bayesian Incentive Compatible Mechanisms.} 

In the multi-buyer setting, we will use $\D_i$ to denote the value distribution of buyer $i$ and $\D=\D_1\times\cdots\times\D_n$ to denote the joint value distribution of all $n$ buyers. Note that each buyer draws his value independently from his respective distribution. For $i\in [n]$, $v_i$ denotes the value function of buyer $i$, $v$ denotes the joint value function of all $n$ buyers, and $v_{-i}$ the joint value function of all buyers other than $i$.

A multi-buyer mechanism maps the joint value function $v$ to an outcome $\all(v)$ and a payment $\pay(v)$ where $\all_i(v)$ denotes the set of items allocated to buyer $i$ and $\pay_i(v)$ denotes the corresponding payment made by buyer $i$. A mechanism $(\all,\pay)$ is {\em incentive compatible} if no buyer can benefit from misreporting his value regardless of others' values. In other words,  if for all $i$, $v_i$, $v'_i$, $v_{-i}$, it holds that 
\[v_i(\all_i(v_i,v_{-i}))-\pay_i(v_i, v_{-i}) \ge v_i(\all_i(v'_i,v_{-i}))-\pay_i(v'_i, v_{-i}) \]
A {\em Bayesian incentive compatible} (BIC) mechanism $(\all,\pay)$ is a pair for which the above inequality holds for all $i$, $v_i$, and $v'_i$ in expectation over $v_{-i}\sim\D_{-i}$. The revenue of a multi-buyer BIC mechanism $M=(\all,\pi)$ is given by $\rev(\D,M):=\expect_{v\sim\D}[\sum_{i\in [n]} \pi_i(v)]$. We use $\rev(\D) := \max_{\text{BIC mechanism } M} \rev(\D,M)$ to denote the optimal revenue achievable by any BIC mechanism over the joint distribution $\D$.

\paragraph{Single-Buyer Mechanisms as Pricings.} From the viewpoint of a single buyer $i$, any given mechanism $M=(\all,\pay)$ generates as a function of the buyer's valuation $v_i$ a distribution over outcomes $(\all(v_i,\cdot), \pi(v_i,\cdot))$. In the following discussion of single parameter mechanisms, where clear from context, we will drop the subscript $i$ to simplify notation. By the \textit{taxation principle}, we can interpret this single buyer mechanism as a (possibly random) price menu $\price :\Delta^{2^{[m]}}\rightarrow \mathbb{R}_+$ that maps each distribution\footnote{For a groundset $X$, we use the notation $\Delta^X$ to denote the set of all probability distributions over $X$.} over subsets of $[m]$, a.k.a. a lottery, to a price the buyer must pay to receive it. We will assume without loss of generality that $\price(\emptyset)=0$ with probability $1$. Under this price menu, the buyer's utility for a (random) set or lottery $S\subseteq [m]$ of items is given by 
\[
    \util(v,p,S) \defeq \operatorname{E}_S\bigg[ v(S) - p(S)\bigg].
\]
The buyer purchases their utility maximizing random set, and the seller collects the associated revenue:
\begin{align*}
    \alloc(v,p) &\defeq \argmax_{S\in\Delta^{2^{[m]}}}\util(v,p,S), \\
    \rev(v,p) &\defeq p\bigg( \alloc(v,p) \bigg).
\end{align*}
In the case of ties, we assume that the buyer purchases any set of maximal price from among the ones that maximize their utility. For simplicity we will overload notation and also let $\alloc(v,p)$ denote the indicator vector of the set purchased by the buyer,
and further abuse notation by calling the average such indicator vector the \emph{expected allocation}, and denoting it and the expected revenue of the mechanism $\price$ by 
\begin{align*}
    \alloc (\D,\price) &\defeq \expect_{p,v\sim \D}[\alloc(v,p)],\\
    \rev(\D,\price) &\defeq  \expect_{p,v\sim \D}[\rev (v,p)].
\end{align*}
Note that the $j^{th}$ coordinate of  $\alloc (\D,\price)$ is the probability that a buyer with valuation $v\sim\D$ purchases item $j$ under pricing $\price$; we will denote this probability by $\alloc_j(\D,\price)$.

\paragraph{Item Pricings and Sequential Item Pricings. } For much of this work we will focus on \emph{item pricing} mechanisms. For a single buyer, we represent an item pricing as a (potentially random) vector $\price\in \mathbb{R}_+^{m}$ mapping individual items to prices, and extend it to sets additively by letting $\price(S)\defeq \sum_{j\in S} p_j$ for $S\subseteq [m]$. For a single buyer, we denote by $\itempricing$ the class of all randomized item pricing mechanisms, meaning that $p\in\R^{m}_+$ may be chosen at random.

In the multi-buyer setting we will further consider {\em Sequential Item Pricing} mechanisms. These mechanisms interact with buyers in a fixed order\footnote{Our results work for an arbitrary given order over buyers, although in some contexts it may be beneficial for the seller to choose a particular order.}. Let $S_i$ denote the set of items left over after the mechanism has interacted with buyers $1$ through $i-1$. When buyer $i$ arrives, the mechanism presents an item pricing over items $S_i$. The buyer instantiates his value and purchases his favorite bundle of items under the pricing, and then the mechanism moves on to the next buyer. The pricing presented to buyer $i$ may depend on the instantiations of values of buyers $1$ through $i-1$ as well as the set $S_i$, but does not depend on the values of buyers $i\cdots n$. 

\paragraph{Buy Many Mechanisms.} In our results, we will also briefly refer to single-buyer \emph{buy-many} mechanisms as defined by \textcite{chawla2022buy}. Informally, a price menu $p$ is buy-many if for every adaptive buyer strategy for sequentially purchasing a sequence of lotteries, there exists a single lottery in the menu that is cheaper and allocates more items to the buyer. We refer the reader to \cite{chawla2022buy} for a formal definition.

%% file: sec_monotone.tex
\section{General Monotone Valuations}\label{sec:monotone}
In this section, we focus on general monotone valuation functions. We propose a sequential item pricing to achieve an $O(\min\{n, m^2\})$ approximation to ex-ante item pricing revenue, and we show that there are no item pricing mechanisms that achieve a revenue of $o(\min\{n,\sqrt{m}\})$ of ex-ante revenue.

\subsection{Upper Bound}

\begin{theorem}\label{thm:monotone-ub} Let $\D$ be any joint distribution for $n$ buyers and $m$ items over general monotone valuation functions. Then 
 \[\earev(\D,\itempricing) \leq \min\{n, 4m^2\} \cdot \rev(\D, \seqitempricing)\]
\end{theorem}

\begin{proof}
    Let $(x^*,p^*)$ be the optimal solutions achieving $\earev(\D,\itempricing)$, with $x^*_i$ denoting the ex ante allocation constraint for buyer $i$ and $p^*_i$ denoting the corresponding distribution over item pricings. For convenience, let \[\obj =\earev(\D,\itempricing) \defeq \sum_{i=1}^n \rev (\D_i, p^*_i).\]

    We will first show how to obtain the $O(n)$ upper bound. This is in fact trivial. We first sample a buyer $i\in [n]$ uniformly at random, and then present them with an item pricing that is instantiated from $p^*_i$, while ignoring all other buyers $i'\neq i$. This will recover a revenue of precisely $\rev (\D_i, p^*_i)$ and then the proof is completed from the fact that we chose buyer $i$ uniformly at random and thus we clearly obtain a $n$-approximation.
    
    Up next, we show how to obtain a $O(m^2)$ upper bound. Recall that since we are restricted to item pricings,
    \[\obj =  \xpectover{p^*}{\sum_{i=1}^n\sum_{j=1}^m p_{ij}\cdot \alloc_j(\D_i,p_i)} ,\]
    and thus, there exists some item $j^*\in [m]$ for which 
    \[\xpectover{p^*}{ \sum_{i=1}^n p_{ij^*}\cdot \alloc_{j^*}(\D_i,p_i)} \geq \frac{\obj}{m}\]

    For a fixed buyer $i$, consider the (random) uniform item pricing that sets the price of \textit{each} item $j\in [m]$ to $p_{ij^*}/m$, given a price $p_i$ that is instantiated from $p^*_i$. Observe that for any realization $v_i \sim \D_i$ where the ex ante buyer allocates to item $j^*$ under pricing $p_i$, the bundle being allocated (and by monotonicity, also the set of all items $[m]$) has value at least $p_{ij^*}$. Therefore, since under price vector $q_i$ the cost for buying all the items is precisely $p_{ij^*}$, buyer $v_i$ will buy at least one item (assuming that all items are available), as it would rather buy the entire bundle rather than buy nothing. 
    
    In effect, the probability of buyer $i$ buying at least one item under pricing $q_i$ is $\lambda_i \geq \alloc_{j^*}(\D_i, p_i)$ whenever all the items are available. Therefore, by rejecting to sell anything (via setting $q_i$ to be infinity) with probability $1-\frac{\alloc_{j^*}(\D_i, p_i)}{2\lambda_i}$, we can ensure that the buyer $i$ buys at least one item with probability $\frac{\alloc_{j^*}(\D_i, p_i)}{2}$, netting a revenue of at least $p_{ij^*}/m$ when doing so.

    Finally, we go through each buyer $i$ one-by-one, selling at price $q_i$ with rejection as defined if all $m$ items are available, and not selling anything if some item has already been sold. Notice that since we sample $p_i$ from $p^*_i$, the probability that a buyer $i$ doesn't purchase anything is precisely $x^*_{ij^*}/2$ under our pricing. To complete the argument, it remains to argue that whenever each buyer $i\in [n]$ arrives, there is a constant probability that all the items are available. As $\sum_{i=1}^n x^*_{ij^*} \le 1$, by an OCRS argument similar to~\Cref{t:ocrs_argument}, all items are available with probability at least $1/2$ upon reaching any buyer; hence, the total revenue extracted from this sequential item pricing is at least
    \[\expect_{p\sim p^*}\big[ \sum_{i=1}^n \frac{1}{2} \cdot \frac{p_{ij^*}}{m} \cdot \frac{\alloc_{j^*}(\D_i, p_i)}{2} \big]\ge \frac{\obj}{4m^2}.\]

    The proof is completed by comparing $n$ to $4m^2$ and picking the better sequential item pricing out of the two.
\end{proof}

\subsection{Lower Bound}

We now present an instance that shows one cannot achieve an approximation of factor better than $\min\{n,\sqrt{m}\}$ of ex-ante revenue through sequential item pricing for general monotone valuations.
\begin{theorem}\label{thm:monotone-lb}
    There exists a joint distribution $\D$ for $n$ buyers and $m$ items over general monotone valuations for which
\begin{align*}
    \earev(\D,\itempricing) &\geq \Omega\left(\min\{n, \sqrt{m}\} \right) \cdot \rev(\D, \seqitempricing).
\end{align*}
\end{theorem}

As with the $\sqrt{\log m}$ lower bound for XOS valuations in \Cref{sec:xos-lb}, this lower bound on the gap between ex ante and sequential item pricing for the class of monotone valuations is more about the difficulty of finding an integral partition of the items that satisfies a typical realization of the buyers' valuation functions than it is about doing so in a sequential, item pricing setting. In particular, the family of instances which demonstrate \Cref{thm:monotone-lb} will show the same gap for $\expostitempricing$ as well, which is an even stronger separation.

This gap is even more precipitous in the sense that it is purely combinatorial, and does not rely on the strategic limitations of mechanisms at all. In particular, this gap holds even for a ``clairvoyant'' mechanism which is given direct access to the realizations $v_i \sim \D_i$ and can extract the full surplus of the buyers' valuations of their allocations as revenue. 

Before constructing our instance, we will introduce the key property which we require of it.

\begin{definition}
    Let $m=\ell^2$ be the number of items. Let $P=(B_0,\ldots,B_{\ell - 1})$ be a partition of these $m$ items to $\ell$ groups of $\ell$ items each, so that we have $|B_i|=\ell$ and $B_i\cap B_j=\emptyset$ for all $i \neq j$. Let $\mathcal{P}$ be a collection of such partitionings. We call a collection $\mathcal{P}$ a \textit{good collection} if for every two partitionings $P_a=(A_0,\ldots, A_{\ell-1}),P_b=(B_0\ldots, B_{\ell-1})\in\mathcal{P}$, we have that $A_i\cap B_j\neq \emptyset$ for all $i, j \in [\ell]$.
\end{definition}
In other words, a good collection is one that for any two partitions $P_a$ and $P_b$ and any two bundles from them, the intersection of the bundles is non-empty. We would like to assign to each buyer such partition, so that if a buyer purchases one bundle, other buyers would not be incentivized to buy any items. 

\begin{lemma}\label{lemma:collection construction}
    For any prime number $\ell$ and $m=\ell^2$, there exists a \textit{good collection} of size $\ell$.
\end{lemma}
\begin{proof}
Set the $\ell^2$ items in a $\ell\times\ell$ matrix, such that each item is indexed by a tuple of form $(i,j)$ where $0\leq i,j \leq \ell-1$. We define partition $P_i=(B_{i0},\ldots, B_{i(\ell-1)})$ as follows: 
   \[B_{ij} = \{ (x,y): y = (xi + j) \bmod{\ell}\},\]
where $(x,y)$ here is the index of the corresponding items. Now first note that for each $i\in \{0, 1, \ldots, \ell - 1\}$, each item $(x,y)$ belongs to exactly one bundle $B_{ij}$, and each bundle consists of exactly $\ell$ items. So the defined $P_i$'s indeed create a collection of partitions. 

Next, we need to prove that for any $i,i',j,j'\in\{0,1,\dotsc, \ell-1\}$ with $i\neq i'$, bundles $B_{ij}$ and $B_{i'j'}$ intersect. Equivalently, we need an item $(x,y)$ for which \[y \equiv xi + j \equiv x i' + j' \pmod{\ell}\]
As $\mathbb{Z}_\ell$ is a field due to $\ell$ being prime, combined with $i' - i \neq 0$, there exists some $x \in \{0, 1, \dots, \ell - 1\}$ where $x(i' - i) \equiv j' - j \pmod{\ell}$. Setting $y = (xi + j) \bmod{\ell}$ concludes the proof.
\end{proof}
\begin{lemma}\label{lemma:Bertrand}
    For any positive integer $m$, there exists a \textit{good collection} of size at least $\ell=\frac{\lfloor\sqrt{m}\rfloor}{2}$ that uses $\ell^2$ items. 
\end{lemma}
\begin{proof}
Due to Bertrand's postulate, we know that for any $n>3$, there exists at least one prime number $p$ such that $2n-2>p>n$. Let $n=\frac{\lfloor \sqrt{m}\rfloor+2}{2}$. Therefore, there exists a prime number $p$ such that $\lfloor \sqrt{m}\rfloor>p>\frac{\lfloor \sqrt{m}\rfloor+2}{2}$. Now using \Cref{lemma:collection construction}, we know that we can build a collection of size $p$ since $p$ is a prime, and we know that $p>\frac{\lfloor \sqrt{m}\rfloor+2}{2}>\frac{\lfloor \sqrt{m}\rfloor}{2}$. This concludes the proof.
\end{proof}

\begin{proof}[Proof of \Cref{thm:monotone-lb}]
Let $\ell$ be the closest prime number no more than $\sqrt{m}$ and let $N=min(n,\ell)$. We prove the theorem in 3 parts: First, we show how to build the instance's valuation function. Then, we show that no ex-post mechanism can gain revenue of more than $1$ under this valuation. Finally, we show that there exists an ex-ante pricing under which the revenue is at least $N(1-\epsilon)$. This concludes the proof of theorem.

\paragraph{Instance.} We index both the items and buyers from $0$. Let $\mathcal{P}=\{P_0,\ldots,P_{\ell-1}\}$ be the \textit{good collection} on the first $\ell^2$ items from \Cref{lemma:collection construction}. For $P_i=(B_{i1},\ldots, B_{i(\ell-1)s})$, let the valuation of buyer $i$ on the $m$ items be deterministic and as follows:
  \[v_i(S) = \mathbbm{1}\bigg[B_{ib}\subseteq S\;\;\text{for some } b\in \{0, 1, \dots, \ell - 1\}\bigg]\text{ for }0\leq i\leq \ell-1, \quad v_i(S)=0 \text{ for }i\geq N\] 

   In other words, we have set zero value for all but the first $\ell^2$ items, and any buyer with index $i \ge N$ has zero valuation on any set of items. Moreover, buyer $i\in \{0, 1, \ldots, N - 1\}$ values any set containing one of their bundles by $1$, and otherwise $0$. This is a monotone valuation, but not sub-additive.\footnote{Let $B_{ib}$ be a bundle of buyer $i$. Then for any $S\subsetneq B_{ib}, |S|>0$ we have that $v_i(S)+v_i(B_{ib} \setminus S)=0<v_i(B_{ib})=1$.}
   
\paragraph{Upper bound on revenue of a mechanism.} Using the property of the \textit{good collection}, if an ex-post mechanism sells a set $S$ containing one of their bundles to a buyer, all the other buyers will value the remaining sets of items by zero. This is because the set $S$ contains at least one item from any bundle of any of the other buyers, and a buyer values a set by zero if their set does not include one of their whole bundles. This means that any mechanism---sequential item pricing mechanisms included---can sell to at most one buyer. Since the maximum social welfare of a buyer is $1$, the maximum revenue achievable will also be at most 1, and can be gained by selling one bundle to a buyer.

\paragraph{Lower bound on ex ante revenue.}
Up next, consider the following mechanism. For each of the buyers $i\in \{0, 1, \ldots, N - 1\}$, pick one of their bundles uniformly at random. This will happen with probability $\frac{1}{\ell}$ for each bundle. Price all of the items in this bundle at  $\frac{1-\epsilon}{\ell}$ for some $\epsilon>0$, and all other items at infinity. In this case, the buyer purchases the randomly selected bundle, and gains a utility of $1-\ell\cdot\frac{1-\epsilon}{\ell} = \epsilon > 0$. The revenue from this draw and this buyer will be $1-\epsilon$. Now note that each item is sold to each buyer with probability $1/\ell$ (i.e. the probability that its corresponding bundle is selected) and so the total allocation of item $j$ is $N/\ell\leq 1$. Therefore, the ex-ante constraint is satisfied, and this mechanism gets a total revenue of $N\cdot (1-\epsilon)$.

We can conclude that for this instance, no ex-post mechanism can get a revenue of better than a factor of $N$ of ex-ante revenue. Now using \Cref{lemma:Bertrand}, we know that $\ell=\Omega(\sqrt{m})$. This concludes the bound in the theorem statement. 
\end{proof}

%% file: sec_gs_proof.tex
\section{Gross Substitutes Valuations}
\label{sec:gs-improvement-main}
The analysis presented in the main part of this work establishes the existence of a sequential item pricing mechanism that achieves a $4e/(e-1)$-approximation to $\earev(\D,\itempricing)$ whenever $\D$ is a distribution over gross substitutes (henceforth, GS) valuations. This was achieved through the design of a $1$-RRS for GS valuations (\Cref{prop:GS-RRS-result}), which in turn implies a $e/(e-1)$-OCRS (\Cref{t:rrs_argument}) and finally a $4e/(e-1)$-approximation of $\earev(\D,\itempricing)$ (\Cref{t:ocrs_argument}). 

However, the previous work of~\textcite{chawla23buy} established a $2$-approximation for \textit{unit demand} valuations (which are also GS), and furthermore this constant is easily shown to be tight for unit-demand valuations even for $m=1$ items via standard prophet inequality bounds. In this section, we will show that one can obtain this factor of $2$ even for GS valuations, proving \Cref{thm:gross-subs} which we restate for convenience.

\gstheorem*

In Section~\ref{sec:gs-improvement} we argue that the analysis of $\chalg$ can be tightened under GS valuations to show that we can actually turn the $1$-RRS for GS valuations into a $1$-OCRS, removing the loss of $e/(e-1)$. This already implies a $4$-approximation for $\earev(\D,\itempricing)$, through~\Cref{t:ocrs_argument}. Unfortunately, we cannot further reduce the approximation constant, as it is inherited from the generality with which we have stated our definitions, in order to capture harder families of valuations.

In Section~\ref{sec:gs-wine-proof} we replicate the analysis of \textcite{chawla23buy}, in order to obtain a $2$-approximation for $\earev(\D,\itempricing)$ under any GS valuation $\D$ and prove \Cref{thm:gross-subs}. While the proof in \cite{chawla23buy} was only shown for unit demand valuations, we observe that the same arguments also hold for GS valuations. We comment that while the idea in this proof is similar to the ones used in this paper, it heavily relies on the GS property and thus, cannot be directly applied to the more general families of valuations that this work addresses.

\subsection{Analysis of $\chalg$ for GS valuations}\label{sec:gs-improvement}

In this section, we show that the $(1-1/e)$-loss in revenue from the analysis of $\chalg$ in Section~\ref{sec:rrs_to_ocrs} is actually not necessary under GS valuations. First, recall that when we defined our $1$-RRS for GS valuations, we simply set $q\defeq p$ as the GS condition handles the revenue constraint (see the paragraph above~\Cref{prop:GS-RRS-result} for the full argument). Let the $\{y^T\}$ and $w$ vectors be defined exactly as in the application of $\chalg$ in Section~\ref{sec:rrs_to_ocrs}. Then, observe that by the definition of the $\{y^T\}$ vectors we have that for any set $Q\subseteq [k]$ and any $j\in Q$:
\begin{equation}\label{eq:gs-inequality}
    y^Q_j = \alpha\cdot q_j^Q \alloc_j (\D_{|Q}, q^Q) 
    = p_j\cdot \alloc_j (\D_{|Q}, p) 
    \geq p_j\cdot\alloc_j (\D,p) 
    = w_j
\end{equation}
where the inequality crucially uses the fact that $\D$ is GS, and thus the allocation of an item under any restriction (that includes it) will never decrease. This inequality is all we need in order to provide a better analysis of $\chalg$ that doesn't suffer the $(1-1/e)$-loss.

Specifically, we proceed to analyze $\chalg$ exactly as in the proof of~\Cref{l:convex_hull_sampler_new}, with the only difference that instead of proving that
\[\sum_{T \subseteq S} \lambda_T \cdot |y^T|\geq \left(1 - \frac{1}{e}\right) \cdot |w| ,\]
we will actually show that
\[\sum_{T \subseteq S} \lambda_T \cdot |y^T| \geq |w| ,\]
from which it becomes clear that the $(1-1/e)$-factor is no longer lost. To prove this, let $s$ be the number of iterations of $\chalg$ and recall that either $\sigma=1$ or $Q_s=\emptyset$ from our termination condition. We will show that~\eqref{eq:gs-inequality} implies that $Q_s=\emptyset$ necessarily. Indeed, assume that this is not the case, and let $j\in Q_s$. Since $Q_s\neq \emptyset$, it must be the case that $\sigma=1$ and thus $\sum_{i=0}^{s-1}\lambda_{Q_i} = 1$. Since $j\in Q_s$, we have that 
\[0<\what_{s,j} \defeq w_j - \sum_{i=0}^{s-1}\lambda_{Q_i}\cdot y^{Q_i}_j \]
which is clearly a contradiction from~\eqref{eq:gs-inequality}, as the sum is simply a linear combination of terms that are at least $w_j$. Thus, we have concluded that $Q_s=\emptyset$ always, which means that $x_s=0$ and thus $w=\sum_{T \subseteq S} \lambda_T \cdot y^T$, which is precisely what we wanted to show.

\subsection{A $2$-approximation for GS valuations}\label{sec:gs-wine-proof}

In this section, we replicate the analysis of \textcite{chawla23buy} to prove \Cref{thm:gross-subs}. The proof is centered around the following lemma, which essentially combines the \textit{Revenue Recovery Scheme} and the \textit{Convex Hull Sampling} steps of our approach.
\begin{lemma}\label{c:gs-wine-guarantte}
    Fix any GS valuation $\D$, any distribution over subsets of items $\Sdist$, any deterministic item pricing $p$, and let $w\defeq \expect_{S\sim\Sdist}[\alloc(\D_{|S}, p)]$. Then, for any vector $y\preceq w$ there exists a randomized item pricing $q=q(\D,\Sdist,p,y)$ such that:
    \begin{enumerate}[(a)]
        \item $\expect_{S\sim\Sdist}[\alloc(\D_{|S},q)] = y$, and,
        \item $\expect_{S\sim\Sdist}[\rev(\D_{|S},q)] = \sum_{j=1}^mp_j\cdot y_j$.
    \end{enumerate}
\end{lemma}
\begin{proof}
    For each set of items $T\subseteq [m]$, we define the vector $x^T\defeq \expect_{S\sim\Sdist}[\alloc(\D_{|S\cap T},p)] $ and observe that $x^T_j = 0$ for any $j\notin T$ and $x^T_j \geq w$ for any $j\in T$, due to the GS property of $\D$. Thus, since $y\preceq w$, we can immediately deduce that vector $y$ lies in the convex hull of the $x^T$ vectors, or equivalently that there exist non-negative coefficients $\{\lambda_T\}_{T\subseteq [m]}$ such that $\sum_{T\subseteq [m]}\lambda_T=1$ and $y = \sum_{T\subseteq [m]}\lambda_T\cdot x^T$.

    The proposed item pricing $q$ samples a set $T\subseteq [m]$ according to the probabilities $\lambda_T$, and then presents the buyer with pricing $q^T$, defined as $q^T_j = p_j$ if $j\in T$ and $q^T_j=\infty$ otherwise. We will now argue that this pricing satisfies the two conditions of the claim.

    For the expected allocation, we have
    \begin{align*}
    \expect_{S\sim\Sdist}[\alloc(\D_{|S},q)] &= \sum_{T\subseteq [m]}\lambda_T \cdot \expect_{S\sim\Sdist}[\alloc(\D_{|S},q^T)]  
    = \sum_{T\subseteq [m]}\lambda_T \cdot \expect_{S\sim\Sdist}[\alloc(\D_{|S\cap T},p)] = y,
    \end{align*}
    with the first equality following from $q$'s definition, the second equality follows from the fact that $q^T$ is equivalent to $p$ under valuations restricted on $T$, and the last equality follows from the definition of $x^T$ and the fact that $y = \sum_{T\subseteq [m]}\lambda_T\cdot x^T$.

    Likewise, for the expected revenue, with similar arguments we have
    \begin{align*}
        \expect_{S\sim\Sdist}[\rev(\D_{|S},q)] = \sum_{T\subseteq [m]}\lambda_T \sum_{j=1}^mp_j   \expect_{S\sim\Sdist}[\alloc_j(\D_{|S\cap T},p)] = \sum_{j=1}^m p_j \cdot\bigg(\sum_{T\subseteq [m]}\lambda_T x_j^T \bigg)= \sum_{j=1}^m p_j \cdot y_j .
    \end{align*}
\end{proof}

\medskip
We are now ready to prove~\Cref{thm:gross-subs}. Consider any GS valuations $\D=\D_1\times\cdots\times\D_n$ and recall that
\[\earev(\D,\itempricing) = \sum_{i=1}^n\earev_{x^*_i}(\D_i,\itempricing)\]
for some ex-ante constraint $\{x^*_i\}_{i=1}^n$ with $\sum_{i=1}^nx^*_{ij}\leq 1$ for all items $j\in [m]$. Consider the sequential item pricing mechanism that offers each buyer $i\in [n]$ the pricing $q_i=q(\D_i,\Sdist_i,p_i,y_i)$ of~\Cref{c:gs-wine-guarantte} over only the available items, for the following instantiations: $\D_i$ is the distribution of valuations for buyer $i$; $\Sdist_i$ is the distribution of available items when $i$ arrives (that depends on the randomness of pricings and allocations over previous buyers); $p_i$ is an item pricing sampled from the distribution that achieves $\earev_{x^*_i}(\D_i,\itempricing)$; and $y_i=(1/2)\cdot \alloc(\D_i,p_i)$.

Define $w_i\defeq \expect_{S_i\sim\Sdist_i}[\alloc(\D_{|S_i},p_i)]$. As long as we can prove that $y_i\preceq w_i$ for all $i\in [n]$, then the guarantees of~\Cref{c:gs-wine-guarantte} will hold, and from condition (b) we will immediately obtain that the revenue the above sequential item pricing collects from each buyer $i$ (call it $\rev_i$) will satisfy
\[\rev_i = \expect_{S_i\sim\Sdist_i, p_i}[\rev(\D_{|S_i},q_i)] = \frac{1}{2}\expect_{S_i\sim\Sdist_i, p_i}\left[\sum_{j=1}^mp_j\cdot\alloc_j(\D_i,p_i)\right] = \frac{1}{2}\cdot \earev_{x^*_i}(\D_i,\itempricing),\]
and thus the proof of~\Cref{thm:gross-subs} will be completed.

In order to argue that $y_i\preceq w_i$, observe that for any item $j\in [m]$ we have
\[w_{ij} = \expect_{S_i\sim\Sdist_i}[\alloc_j(\D_{|S_i},p_i)] = \pr[j\in S_i]\cdot  \expect_{S_i\sim\Sdist_i}[\alloc_j(\D_{|S_i},p_i)\big| j\in S_i] \geq \pr[j\in S_i]\cdot \alloc_j(\D_i,p_i),\]
where the inequality follows from the GS property, under which the allocation of any item under any restriction that includes it can only increase. Thus, a sufficient condition to show that $y_i\preceq w_i$ is to prove that $\pr[j\in S_i]\geq \frac{1}{2}$ for all buyers $i\in [n]$. This will certainly be true for the first buyer, as $S_1=[m]$. Then, from condition (a) of~\Cref{c:gs-wine-guarantte} we obtain that the expected allocation of items to the first buyer is exactly $y_1$. Inductively, we proceed to show that each buyer $i'<i$ gets expected allocation $y_{i'}$. Then, we get that
\[\pr[j\in S_i] = 1 - \sum_{i'<i}y_{i'j} \geq 1 -\sum_{i'=1}^ny_{i'j} = 1 - \frac{1}{2}\sum_{i'=1}^n x^*_{i'j}\geq \frac{1}{2},\]
with the last inequality following from the ex ante constraint. Thus, we have shown via induction that $\pr[j\in S_i]\geq 1/2$ for all $j\in [m]$ and $i\in [n]$, completing the proof.

%% file: sec_rrs_lowbound.tex
\section{A Lower Bound for RRSes on XOS Valuations}
\label{sec:rrs-lb}

In~\Cref{thm:subadd-RRS}, we proved the existence of an $\bigoh(\log m + \log \Gamma)$-RRS, where $\Gamma \defeq p_{max}/p_{min}$ is the aspect ratio of the given price. We now show that this dependency on $\Gamma$ is necessary: without it, we would have to suffer a factor of $\Omega(\sqrt{m})$ in our approximation. This is captured through the following.

\begin{theorem}\label{lem:rrs_lower_bound_sqrtm}
    For all $m$, there exists a distribution $\D$ over XOS valuations, a deterministic item pricing $p$, and a deterministic subset of items $S \subseteq [m]$, such that for all item pricings $q$,
    \[\rev(\D_{|S}, q) \le \frac{1}{\Omega(\sqrt{m})} \cdot \sum_{j \in S} p_j \cdot \alloc_j(\D, p).\]
    This rules out the existence of any $o(\sqrt{m})$-RRS for XOS valuations.
\end{theorem}

Before proceeding, we note that this lemma is not a counterargument to~\Cref{it:ocrs}, as the ex ante allocation constraints for the counterexample are indeed exponential in $m$. Our counterexample is as follows.

\paragraph{Available items and pricings.} We define the subset of items $S \defeq [m - 1]$ (i.e. only the last item is not available), and the pricing vector $p \defeq (\beta, \beta^2, \dots, \beta^{m - 1}, 0)$, where $\beta \defeq \sqrt{m - 1}$.

\paragraph{Valuation distribution.} We define the (deterministic) XOS valuation function $v_{i,R}$ that is parametrized by an index $i \in [m - 1]$ and a set $R \subseteq [i - 1]$. Such a valuation function consists of two additive components $v^1_{i,R}$ and $v^2_{i,R}$, i.e.
\[v_{i,R}(T) \defeq \max\bigg\{ \sum_{j\in T} v^1_{i,R}(j) \;,\; \sum_{j\in T} v^2_{i,R}(j) \bigg\}\]
for any $T\subseteq [m]$, where the additive components are defined as

\begin{align*}
    v^1_{i,R}(j) = \begin{cases}
        \beta^i & \text{if $j = i$,} \\
        \epsilon + \sum_{k \in R} (\beta^i - \beta^k) & \text{if $j = m,$} \\
        0 & \text{otherwise.}
    \end{cases}
    && \text{and} &&
    v^2_{i,R}(j) = \begin{cases}
        \beta^i & \text{if $j \in R$,}\\
        0 & \text{otherwise.}
    \end{cases}
\end{align*}
for some tine $\epsilon > 0$. The buyer valuation distribution $\D$ is then defined via the following sampling procedure:
\begin{enumerate}
\item An index $i \in [m - 1]$ is sampled proportionally to $\beta^{-i}$, i.e. with probability $\frac{\beta^{-i}}{\sigma}$ for $\sigma \defeq \sum_{i=1}^{m-1} \beta^{-i}$.
\item A set $R \subseteq [i - 1]$ is sampled by including each element $j\in [i-1]$ with probability $(m-1)^{-1/2}$ independently.
\item Finally, valuation $v_{i,R}$ is realized.
\end{enumerate}

We now begin the proof by examining the right-hand side of the statement of~\Cref{lem:rrs_lower_bound_sqrtm}. Observe that for any $i \in [m - 1]$ and $R \subseteq [i - 1]$, the utility-maximizing set under the valuation function $v_{i,R}$ and our chosen pricing $p$ is the set $\{i, m\}$; note that this is the utility-maximizing set for the first additive component with utility of $\epsilon + \sum_{k \in R} (\beta^i - \beta^k)$, while the maximum utility from the second additive component is $\sum_{k \in R} \beta^i - \beta^k$. Therefore, we have $\alloc_j(\D, p) = \beta^{-j}/\sigma$ for all $j \in S$ (as whenever we draw $j$ in the sampling process of $\D$, we know automatically that the utility-maximizing set is $\{j, m\}$). Hence,
\[\sum_{j \in S} p_j \cdot \alloc_j(\D, p) = \sum_{j=1}^{m-1} \beta^j \cdot \frac{\beta^{-j}}{\sigma} = \sigma^{-1}\cdot (m-1).\]
We also comment that the allocation probabilities are exponentially small in $m$, and thus they do not satisfy the assumption of~\Cref{it:ocrs}.

To complete the proof, it suffices to show that under any pricing $q$, we have 
\begin{equation}\label{eq:xos-instance-to-show}
\rev(\D_{|S}, q) \le 4 \sigma^{-1} \sqrt{m - 1}.
\end{equation}
Fix any pricing $q$. We define a labelling vector $\ell \in \mathbb{Z}^{m - 1}$, where $\ell_j = \lceil \log_\beta q_j \rceil$ for all $j \in [m - 1]$. We now consider the maximum revenue we can extract from $v_{i,R|S}$ for any $i$ and $R$. Since $v_{i,R}$ is defined as the maximum over two additive valuations, we consider the three following ways via which we can extract revenue from it.

\paragraph{Revenue is extracted from $v^1_{i,R|S}$ and $\ell_i \neq i$.} Since revenue is extracted from $v^1_{i,R|S}$ and $m\notin S$, the only item that can be purchased is $i$, valued at $\beta^i$. Furthermore, since $\ell_i\neq i$, we can conclude that $q_i\leq \beta^{i-1}$ and therefore, the total revenue extracted from this case is at most
\[\operatorname{E}_{i, R} [\beta^{i - 1}] = \sum_{i=1}^{m-1} \beta^{i - 1} \cdot \frac{\beta^{-i}}{\sigma} = \sigma^{-1} \cdot \frac{m - 1}{\beta} = \sigma^{-1} \sqrt{m - 1}.\]

\paragraph{Revenue is extracted from $v^1_{i,R|S}$ and $\ell_i = i$.} In this case, the revenue extracted is at most $\beta^i$. A critical observation here is that we need $\ell_j \ge i$ for all $j \in R$ in order for this case to happen; note that the utility of $S$ under $v^1_{i,R|S}$ is $\beta^i - q_i < \beta^i - \beta^{i - 1}$, since $\ell_i = i$, while if any item $j \in R$ has $\ell_j \le i - 1$, then its utility under $v^2_{i,R|S}$ is $\beta^i - q_j \ge \beta^i - \beta^{i - 1}$, and thus the buyer would deviate away from buying item $i$ and towards buying this $j$ under $v^2_{i,R|S}$.

We now define the set $U = \{i \in S : \ell_i = i\}$, and note that revenue in this case can only be extracted from items in $U$. Furthermore, consider any item $i \in U$ and any $R$, we can only extract revenue from item $i$ under $v^1_{i,R|S}$ if $R \cap U = \emptyset$; otherwise, some item $j \in R$ has $\ell_j = j < i$, and by the previous observation we cannot extract revenue from item $i$ under $v^1_{i,R|S}$. Therefore, if we enumerate $U = \{u_1, u_2, \dots, u_t\}$ where $u_1 < u_2 < \ldots < u_t$, the total revenue extracted in this case is at most
\begin{align*}
    \expect_{i, R}[\beta^{i} \cdot \mathbbm{1}(\text{revenue from $v^1_{i,R|S}$ and $\ell_i = i$})] &\le \expect_{i, R}[\beta^i \cdot \mathbbm{1}[i \in U \wedge R \cap U = \emptyset]] \\
    &= \sum_{j=1}^t \beta^k \pr[i = u_j] \pr[R \cap U = \emptyset \mid i = u_j] \\
    &= \sum_{j=1}^t \beta^{u_j} \cdot \frac{\beta^{-u_j}}{\sigma} \cdot \left(1 - \frac{1}{\sqrt{m - 1}}\right)^{j - 1} \\
    &\le \sigma^{-1} \sum_{j=0}^\infty \left(1 - \frac{1}{\sqrt{m - 1}}\right)^{j} = \sigma^{-1}\sqrt{m - 1}
\end{align*}
where the second line uses the fact that each item $u_1, u_2, \dots, u_{j - 1}$ does not appear in $R$ with probability $1 - \frac{1}{\sqrt{m - 1}}$ independently.

\paragraph{Revenue is extracted from $v^2_{i,R|S}$.} We first fix any item $j$, and for convenience let $l = \ell_j$. Observe that item $j$ can only be sold in this case if $l \le i$ (otherwise when $l > i$, we have $q_j > \beta^{l - 1} \ge \beta^i = v^2_{i,R|S}(j)$); furthermore, we obviously have $q_j \le \beta^{l}$. Therefore, the total revenue extracted from item $j$ in this case is at most

\begin{align*}
    \expect_{i, R}[\beta^l \cdot \mathbbm{1}(\text{revenue from $v^2_{i,R|S}$ and $j$ is included})] &\le \expect_{i, R}[\beta^l \cdot \mathbbm{1}[j \in R \wedge i \ge l]] \\
    &= \beta^l \sum_{k=1}^{m - 1} \pr[i = k] \pr[j \in R \mid i = k] \\
    &\le \frac{\beta^l}{\sqrt{m - 1}} \sum_{k=l}^{m - 1} \frac{\beta^{-k}}{\sigma} \\
    &\le \frac{\beta^l}{\sqrt{m - 1}} \cdot \frac{2 \beta^{-l}}{\sigma} = \frac{2\sigma^{-1}}{\sqrt{m - 1}}
\end{align*}
where the second line uses that $\pr[j \in R \mid i = k]$ is $\frac{1}{\sqrt{m - 1}}$ if $i > j$, or $0$ otherwise; the second line assumes that $\beta = \sqrt{m - 1} \ge 2$, which means we can bound $\sum_{k=l}^{m - 1} \beta^{-k} \le 2 \beta^{-l}$.

Thus, the total revenue extracted from this case is at most 
\[\expect_{i, R}\left[\sum_{j=1}^m \beta^{\ell_j} \cdot \mathbbm{1}(\text{revenue from $v^2_{i,R|S}$ and $j$ is included})\right] \le 2 \sigma^{-1} \sqrt{m - 1}.\]

Combining the three cases gives us an upper bound of $4 \sigma^{-1} \sqrt{m - 1}$ on $\rev(\D_{|S}, q)$ for any $q$.

%% file: sec_xos_lb_proofs.tex
\section{Deferred Proofs from \Cref{sec:xos-lb}} \label{sec:xos-lb-proofs}
We start by reproducing the description of our instance for convenience. 

\paragraph{The buyers.} Let $n$ buyers $v_i \sim \D_i$ be (independently) identically distributed, with valuations sampled according to the following process:
\begin{enumerate}
    \item A set of items $A\subseteq [m]$ with $|A|=k$ is sampled uniformly at random.
    \item An integer $h \in \{1,2,\dotsc , \frac{1}{2}\log k \}$ is drawn uniformly at random. We denote $\ell\defeq2^h$.
    \item The valuation of the buyer is then realized as an XOS function, which is given by the maximum over the following additive valuations:
    \begin{enumerate}
        \item For the set $A$, we define an additive valuation $v_A$ such that $v_A(j)= 1 + t$ for items $j\in A$ and $v_A(j)=0$ for items $j\notin A$.
        \item For \textit{every} set $B$ of cardinality $|B| = t\cdot \ell$, we define an additive valuation $v_B$ such that $v_B(j)=1 + \frac{k}{\ell}$ for items $j\in B$ and $v_B(i)=0$ for items $j\notin B$. 
    \end{enumerate}
\end{enumerate}
We will take $n = k = \sqrt{m}$ and $t = \sqrt{\log k}$. Observe that for our selected parameters $t\cdot \ell \leq \sqrt{m\cdot\log m}\leq m$, and all parameters are integer-valued, so this is a well-defined instance. Let $A_i$ be the $A$-set for buyer $i$.

\bigskip
In order to prove the instance-specific claims in \Cref{sec:xos-lb}, we will require more notation. We use $\rev(v_i,p)$ to denote the revenue that we collect from the $i$-th buyer if their valuation is realized to $v_i$ and the pricing is some fixed vector $p\in\Rplus^m$. Note that this is a deterministic quantity, as all sources of randomness are fixed.

For our instance we will define the events $\E^A(v_i,p)$ and $\E^B(v_i,p)$ that subset of items $[m]$ that buyer $i$ with valuation $v_i$ purchases is of type $A$ or $B$ respectively; more formally, $\E^A(v_i,p)$ holds if for prices $p$ the buyers' utility for their preferred bundle $T\subseteq [m]$ is given by 
\[
    u_i(T) = v_i(T) - p(T) = \sum_{j \in T} (v_A(j) - p_j),
\]
and otherwise $\E^B(v_i,p)$ holds if their valuation is given by some $v_B$.

\xosfirstlemma*

\begin{proof}[Proof of \Cref{c:xos-lb-A-items}]
    We consider a fixed buyer $v = v_i$ (and thus also the set $A$ and the integer $\ell$ in the valuation description) and a fixed pricing $p$. Since $\E^A(v,p)$ holds, we know that the buyer will purchase a set of items such that their utility is maximized through the linear valuation $v_A$; in particular, including a subset $A_T\subseteq A$ given by
    \[
        A_T \defeq \argmax_{T\subseteq A}(v_A(T)- p(T)) = \{j\in A : p_j \leq t+1\}.
    \]
    We will show that $|A_T|\geq \frac{kt}{t+1}$. To that end, let $B_{AT} \subseteq A_T$ be a maximal $B$-utility subset of $A_T$. Observe that since $t\leq \frac{k}{\ell}$ for our chosen parameters, every item worth buying under an $A$-valuation $v_A$ will also be worth buying under some $B$-valuation $v_B$. Thus, $B_{AT}$ is any maximal-cardinality subset of $A_T$ of size at most $t \ell$.
    
    Since $\E^A(v_i,p)$ holds we know that $i$ prefers $A_T$ to $B_{AT}$, meaning $A_T$ has higher utility than $B_{AT}$, so
    \begin{equation} \label{eq:AT-vs-BAT-utility}
        \sum_{j\in A_T}(t + 1 - p_j)\geq \sum_{j\in B_{AT}}\left(1 + \frac{k}{\ell}-p_j\right).
    \end{equation}
    We first argue that $|A_T| > t\ell$. Otherwise $B_{AT}=A_T$, and thus \eqref{eq:AT-vs-BAT-utility} can only be satisfied if $t\geq \frac{k}{\ell}$ which is not the case as $t\ell \leq \sqrt{\log k}\cdot \sqrt{k} < k$. Therefore $|A_T|\geq t\ell$. In this case, $|B_{AT}|= t\ell$ and \eqref{eq:AT-vs-BAT-utility} can be restated as
    \begin{equation} \label{eq:AT-vs-BAT-utility-consequence}
        \sum_{j\in A_T\setminus B_{AT}} (t + 1 - p_j) \geq \sum_{j\in B_{AT}}\left(\frac{k}{\ell}-t\right) = kt - t^2\ell,
    \end{equation}
    and since $p_j\geq 0$ for all $j$ this implies that
    \[
        (t + 1)\cdot (|A_T| - t\ell) \geq kt - t^2\ell.
    \]
    This immediately implies that $|A_T|\geq \frac{kt}{t+1}$, as claimed.
\end{proof}

\xossecondlemma*

\begin{proof}[Proof of \Cref{c:xos-lb-A-rev}]
    We will prove this claim for every item pricing $p$ and every buyer valuation $v = v_i \sim \D_i$, provided that the event $\E^A(v_i, p)$ holds. The claim will then follow from taking the conditional expectation over $v_i$.
    
    We use the same setup as in the proof of~\Cref{c:xos-lb-A-items} to define the sets $A_T$ and $B_{AT}$. Since $\E^A(v_i,p)$ holds by assumption, we know that $\rev(v_i,p) = p(A_T)$ because the buyer maximizes their utility via $v_A$ and buys this set. Therefore we will bound $p(A_T)$.
    
    Just as before, \eqref{eq:AT-vs-BAT-utility} implies that $|A_T| \geq t\ell$ and that
    \begin{equation} \label{eq:AT-vs-BAT-utility-consequence-restated}
        \sum_{j\in A_T\setminus B_{AT}} (t + 1 - p_j) \geq kt - t^2\ell
    \end{equation}
    holds for \emph{any} subset $B_{AT}$ of $A_T$ with cardinality $t\ell$; in particular it holds for $B_{AT}$ the $t\ell$ items in $A_T$ of minimum price. For this choice of $B_{AT}$ the average item price in $A_T\setminus B_{AT}$ exceeds the average item price in $A_T$, i.e.
    \[
        \frac{p(A_T\setminus B_{AT})}{|A_T\setminus B_{AT}|}\geq \frac{p(A_T)}{|A_T|}.
    \]
    Plugging this into \eqref{eq:AT-vs-BAT-utility-consequence-restated} and observing that $|A_T\setminus B_{AT}| = |A_T| - t\ell$ yields
    \[
        kt - t^2\ell\leq \sum_{j\in A_T\setminus B_{AT}} (1+t-p_j) = (1+t)(|A_T|-t\ell) - p(A_T \setminus B_{AT}) \leq (|A_T|-t\ell)\cdot \left(1 + t - \frac{p(A_T)}{|A_T|}\right).
    \]
    Solving for $p(A_T)$ we obtain
    \[
        p(A_T) \leq |A_T|\cdot \frac{|A_T|(t+1)- t(k+\ell)}{|A_T|-t\ell},
    \]
    and using that $|A_T|\leq |A| = k$ (by definition) and that $|A_T|\geq \frac{kt}{t+1}$ by \Cref{c:xos-lb-A-items} since $\E_i^A(v,p,S)$ holds, we finally obtain that
    \[
        p(A_T) \leq k\cdot \frac{k(t+1) - t(k+\ell)}{\frac{kt}{t+1}-t\ell}.
    \]
    The proof is completed by showing
    \[
        \frac{k(t+1) - t(k+\ell)}{\frac{kt}{t+1}-t\ell} \leq c_1
    \]
    for some constant $c_1$ and all sufficiently large $k$, which easily follows from observing that $k\geq 2t\ell$ for our choice of parameters $t=\sqrt{\log k}$ and $\ell \in [2, \sqrt{k}]$.
\end{proof}

\xosthirdlemma*

\begin{proof}
    Consider a fixed item pricing $p\in\Rplus^m$. 
    We can relabel the items so that their prices are non-decreasing, i.e. we assume without loss of generality that $p_1\leq p_2\leq \dotsc \leq p_m$.

    Since we are conditioning on event $\E_i^B(v,p,S)$, we know that the buyer always buys a set of items whose utility is maximized through a $B$-valuation, thus they buy a set of at most $t\ell$ items where $\ell$ is a buyer-dependent parameter that is drawn uniformly at random from $L=\{2,4,8,\dotsc, \sqrt{k}\}$. Furthermore, the buyer will clearly prioritize the cheaper items, as all items have the same (additive) valuation of $1+k/\ell$ under $B$ and cheaper items contribute more to the buyer's utility. Thus, for a given $\ell$, the buyer will purchase items $T$ which are a prefix of the items in ascending-price-order:
    \[
        T = \{1,2,3,\dotsc, x\} \subseteq [m],
    \]
    where $x$ is the maximum index for which $x\leq t\ell$ and $p_x \leq 1 + \frac{k}{\ell}$. We will bound the conditional expected revenue, which is precisely the conditional expectation of $p(T)$, by observing that
    \begin{align}
        \xpectover{v_i}{\rev(v_i, p) \mid \E^B(v_i,p)} &= \sum_{\ell' \in L} \xpectover{v_i}{\rev(v_i, p) \mid \E^B(v_i,p) \wedge \ell(v_i) = \ell '} \cdot \probover{v_i}{\ell(v_i) = \ell'} \notag \\
        &= \frac{1}{|L|} \cdot \sum_{\ell' \in L} \xpectover{v_i}{p(T) \mid \E^B(v_i,p) \wedge \ell(v_i) = \ell '}. \label{eq:xos-lb-B-rev-bound-sum}
    \end{align}
    Here we use $\ell(v_i)$ to denote the value of $\ell$ that $v_i \sim \D_i$ realizes, and we used that $\probover{v_i}{\ell(v_i) = \ell'} = \frac{1}{|L|} = \frac{2}{\log k}$ for all $\ell' \in L$. 
    We will bound this sum \eqref{eq:xos-lb-B-rev-bound-sum} by $c_2 \cdot \frac{tk}{\log k}$ for some constant $c_2$ and sufficiently large $k$. We will do this by a combination of case analysis and grouping the possible values of $\ell$ as a function of the non-uniform prices. Let 
    \[
        \ell^* \defeq \max
        \left(\left\{\ell\in L: \:  p_{t\ell} \leq \frac{k}{\ell}+1\right\}\right).
    \]
    This set is nonempty and the threshold parameter value $\ell^*$ exists provided that $p_{2t}\leq \frac{k}{2}+1$. This is because the thresholds $\frac{k}{\ell} + 1$ are decreasing in $\ell$, while the prices $p_{t\ell}$ are increasing in $\ell$. Therefore we have two cases to consider.

    \paragraph{Case $1$.} First, we consider the case that $p_{2t} > \frac{k}{2}+1$. In this case, since $t+1\leq \frac{k}{2}+1$ for $t=\sqrt{\log k}$ and also $\frac{k}{\ell} + 1 \leq \frac{k}{2}+1$ for $\ell\geq 2$, we know the buyer will not buy any item $j$ for $j\geq 2t$; since $p_j$ is too high it would necessarily decrease their utility. Furthermore, the maximum revenue attainable from each of the $2t$ remaining items is at most equal to their value to the buyer, which is $(\frac{k}{\ell}+1)$ for $\ell\in L$ uniformly at random. Thus, from \eqref{eq:xos-lb-B-rev-bound-sum} we obtain
    \[
        \xpectover{v_i}{\rev(v_i, p) \mid \E^B(v_i,p)} \leq 2t\sum_{\ell\in L}\frac{1}{|L|}\cdot \left(1+\frac{k}{\ell}\right) = \frac{1}{|L|}\cdot \left(2t|L| + 2tk \right) \leq \frac{6tk}{\log k},
    \]
    where the last statement follows from $|L|=\frac{1}{2}\log k$ and the fact that $L=\{2,4,8,\dotsc ,\sqrt{k}\}$ follows a geometric progression and so the $\ell^{-1}$ terms sum to at most $1$.

    \paragraph{Case $2$.} We now address the case where $\ell^*$ is well-defined. We consider a buyer $v = v_i$ and again use $\ell(v)$ to denote the $\ell$ value that their valuation realized. We analyze their expected revenue based on whether $\ell(v)<\ell^*$, $\ell(v)=\ell^*$, or $\ell(v)>\ell^*$.

    If $\ell(v) = \ell^*$, then we know that the buyer will purchase at most $t\ell^*$ items, and also that each item contributes at most the buyer's per-item value to the revenue, which is $\frac{k}{\ell^*}+1$. Thus we have 
    \begin{equation} \label{eq:xos-lb-B-rev-lstar}
        \xpectover{v}{\rev(v,p) \mid \E^B(v,p) \wedge \ell(v) = \ell^*} \leq t\ell^* \left(1 + \frac{k}{\ell^*}\right) = t(k+\ell^*).
    \end{equation}
    
    If $\ell(v) < \ell^*$, then each item sold can contribute at most $(\frac{k}{\ell^*} + 1)$ to the revenue, and for each $\ell$ at most $t\ell$ items can be sold. Thus, we have
    \begin{equation} \label{eq:xos-lb-B-rev-lsmall}
        \sum_{\ell\in L : \ell <\ell^*}\xpectover{v}{\rev(v,p) \mid \E^B(v,p) \wedge \ell(v) = \ell} \leq \sum_{\ell \in L: \ell < \ell^*} t\ell\cdot \left(1+\frac{k}{\ell^*}\right) \leq t(k+\ell^*),
    \end{equation}
    where the last inequality follows from $\sum_{\ell\in L: \ell <\ell^*}\ell \leq \ell^*$ again because the terms of $L$ follow a geometric progression.

    Finally, if $\ell(v) > \ell^*$ then each individual item can contribute at most $\frac{k}{\ell(v)}+1$ to the revenue, because this is the buyer's per-item value. And by the definition of $\ell^*$ there can be at most $t\ell^*$ items such that $p_j \leq \frac{k}{\ell(v)}+1 \leq \frac{k}{\ell^*}+1$; only these items may generate revenue. Thus, we have
    \begin{equation} \label{eq:xos-lb-B-rev-lbig}
        \sum_{\ell\in L : \ell >\ell^*}\xpectover{v}{\rev_i(v,p,S)\mid \E_i^B(v,p,S) \wedge \ell(v) = \ell} \leq t\ell^*\cdot\sum_{\ell \in L: \ell > \ell^*}  \left(1+\frac{k}{\ell}\right) \leq t\ell^*\cdot |L| + tk,
    \end{equation}
    where the last inequality once again follows because the terms of $L$ follow a geometric progression, so the sum of the $\ell^{-1}$ terms is at most $1$. 

    Combining everything, applying \eqref{eq:xos-lb-B-rev-lstar}, \eqref{eq:xos-lb-B-rev-lsmall}, and \eqref{eq:xos-lb-B-rev-lbig} to \eqref{eq:xos-lb-B-rev-bound-sum} we have 
    \begin{align*}
        \xpectover{v_i}{\rev(v_i,p) \mid \E^B(v_i,p)} &= \frac{1}{|L|}\cdot \sum_{\ell\in L} \xpectover{v_i}{\rev(v_i,p) \mid \E^B(v_i,p) \wedge \ell(v) = \ell} \\
        &\leq \frac{1}{|L|}\cdot (2t(k+\ell^*) + t\ell^*|L| + tk),
    \end{align*} 
    and since $\ell^*\leq \max(L) = \sqrt{k}$ and $|L| = \frac{1}{2} \log k$, we have that
    \begin{align*}
        \xpectover{v_i}{\rev(v_i,p) \mid \E^B(v_i,p)}\leq c_2 \cdot \frac{tk}{\log k}
    \end{align*} 
    for some constant $c_2$ and all suitably large $k$, as desired. 
    
    This holds in both cases, completing the proof.
\end{proof}

Before proving the \cref{lem:max-t-feasible-num-buyers}, recall the definitions of an \emph{assignment} $\sigma$ of items to buyers and the definition for $I\subseteq [n]$ to be $(1-c)$-\emph{feasible} under $\sigma$ (\Cref{def:item-assn,def:feasible-assn}). We will use the following more general claim:

\begin{lemma}\label{lem:approx-packing}
    Let $\mathcal{E}(c,C,n)$ denote the event that there exist some subset of buyers $I \subseteq [n]$ and assignment $\sigma: [m] \rightarrow [n] \cup \{\bot\}$ such that $|I| \geq Cn$ and $I$ is $(1-c)$-feasible under $\sigma$.
    For any fixed constant $\gamma > 0$, if $c = \gamma \cdot \log^{-1/2}n$ and $C = 8 c = 8 \gamma \log^{-1/2}n$, then
    \[
        \probover{v \sim \D}{\mathcal{E}(c,C,n)} = \bigo{e^{-n^2/\log^{2} n}}.
    \]
\end{lemma}

\begin{proof}[Proof of \Cref{lem:approx-packing}]
    To begin, consider a fixed $I \subseteq [n]$ let $\alpha \defeq |I|/n$. We will be interested in the $I$ for which $\alpha \geq C$. Next let $N(I) \defeq \cup_{i \in I} A_i$ be the union of the $A$-sets for $i\in I$, and let $A_i^\sigma \defeq A_i \cap \sigma^{-1}(i)$, so that $\sigma$ is $(1-c)$-feasible for $i$ precisely when $|A_i^\sigma| \geq (1-c)|A_i| = (1-c)n$. 
    
    Next observe that for all assignments $\sigma$, 
    \begin{equation*}
        \min_{i \in I} |A_i^{\sigma}| \leq \frac{1}{|I|} \cdot \sum_{i \in I} |A_i^{\sigma}| \leq \frac{1}{|I|} \cdot |N(I)|,
    \end{equation*}
    since $A_i^\sigma \subseteq A_i$ by definition and the $A_i^\sigma$ are disjoint because $\sigma$ is a well-defined mapping. Therefore if $\sigma$ is $(1-c)$-feasible for $I$, then
    \begin{equation} \label{eq:I-feasible-relaxation}
        |I| n(1-c) \leq |I| \min_{i \in I} |A_i^{\sigma}| \leq |N(I)|.
    \end{equation}
    We will argue that for fixed $I$ a $(1-c)$-feasible assignment for $I$ is unlikely to exist, because this would imply that $|N(I)| \geq |I| (1-c)n$ by \eqref{eq:I-feasible-relaxation}, which we will show is unlikely.

    To this end, let $X_{ij} \defeq \mathbbm{1}\{j \in A_i\}$ indicate the event that $j\in A_i$; then
    let $Y_j \defeq \max_{i \in I} X_{ij}$ indicate that $j$ is in at least one set $A_i$ for $i \in I$. Then observe that $Y \defeq \sum_j Y_j = |N(I)|$.
    
    \begin{claim} \label{lem:Yj-are-NA}
        $\{Y_j\}$ is negatively associated (NA).
    \end{claim}
    \begin{proof}
    This follows from results about NA variables by \textcite{joag1983negative}; in particular closure and composition properties for collections of NA variables.
    
    First, consider each collection of indicator random variables $\mathcal{X}_i \defeq \{X_{i'j} : i' = i\}$. Exactly $n$ of these are 1, indicating the subset $A$ chosen uniformly at random from $[m]$; therefore this is a permutation distribution, and so by \cite{joag1983negative} Theorem 2.11 this collection $\mathcal{X}_i$ is NA. Next, consider the larger collection $\mathcal{X} \defeq \{{} X_{ij}\}$; as the union of independent collections of NA variables, $\mathcal{X}$ is also NA (\cite{joag1983negative}, Property $P_7$). Finally, consider the collection $\mathcal{Y} \defeq \{Y_j\}$. Since $Y_j \defeq \max_{i \in I} X_{ij}$, each $Y_j$ is defined on a unique subset of the NA variables $\mathcal{X}$; for any $Y_j$ and $Y_{j'}$, the sets of $X_{ij}$ upon which they depend are disjiont. Since $\max(\cdot)$ is an increasing function, therefore $\mathcal{Y}$ is NA also (\cite{joag1983negative}, Property $P_6$).
    \end{proof}

    \medskip 
    \noindent \textbf{Bounding $|N(I)|$:}
    Since \Cref{lem:Yj-are-NA} demonstrates that $\mathcal{Y}$ is negatively associated, we can use concentration. First note that for all $j$,
    \[
        \xpectover{v\sim \D}{Y_j} = 1 - \probover{v\sim\D}{\bigwedge_{i \in I} (j \not\in A_i)} = 1 - (1-1/n)^{|I|} = 1 - (1-1/n)^{\alpha n}.
    \]
    Then we may bound the probability that $|N(I)| \geq |I|(1-c)n$ by applying a Hoeffding bound to $\sum_j Y_j$ for any $I$ with $|I|\geq Cn$ as follows:
    \begin{align} 
        \probover{v\sim\D}{|N(I)| \geq \alpha n^2 (1-c) } &= \probover{v\sim\D}{Y - \xpectover{v\sim\D}{Y} \geq \alpha n^2 (1-c) - \xpectover{v\sim\D}{Y} }\notag \\
        &\leq \exp\left(\frac{-2\left(\alpha n^2 (1-c) - \xpectover{v\sim\D}{Y}\right)^2}{n^2}\right) \notag \\
        &= \exp\left(\frac{-2 \left(\alpha n^2 (1-c) - n^2(1 - (1-1/n)^{\alpha n})\right)^2}{n^2}\right). \notag 
        \intertext{Since $1-(1-1/n)^{\alpha n} \leq 1 - e^{-\alpha} (1-O(1/n))^\alpha \leq 1 - e^{-\alpha} + e^{-\alpha}\cdot O(1/n) \leq \alpha - \alpha^2/2 +\alpha^3/6  + e^{-\alpha}\cdot O(1/n)$, therefore}
        \probover{v\sim\D}{|N(I)| \geq \alpha n^2 (1-c) } &\leq \exp\left(-2 \alpha^2 n^2 \left( (1-c) - (1 - \frac{\alpha}{2} +\frac{\alpha^2}{6} + \frac{e^{-\alpha}}{\alpha}O(n^{-1})\right)^2\right) \notag \\
        &= \exp\left(-2 \alpha^2 n^2 \left(\frac{\alpha}{2} -\frac{\alpha^2}{6} - \frac{e^{-\alpha}}{\alpha}O(n^{-1}) - c\right)^2\right) \notag \\
        &\leq \exp\left(-\frac{1}{8} n^2 \alpha^4 \left(1 -\frac{2\alpha}{3} - 4\frac{e^{-\alpha}}{\alpha^2}O(n^{-1})\right)^2\right), \notag
        \intertext{since $\alpha \geq C > 4c$, and since $\alpha \geq C \gg n^{-1/3}$ for sufficiently large $n$ this last term is $o(\alpha)$, so}
        &\leq \exp\left(-\frac{1}{8} n^2 \alpha^4 \left(1 -\frac{3\alpha}{4}\right)^2\right) \notag  \\
        &\leq \exp\left(-\frac{9}{128} n^2 \alpha^4 \right), \notag
    \end{align}
    where we used that $1-3\alpha/4 \geq 3/4$. Since $I$ is a set for which $\alpha \geq C$, by \eqref{eq:I-feasible-relaxation}, we therefore know that 
    \begin{equation} \label{eq:feasible-assn-bound}
        \probover{v\sim\D}{\exists \sigma: \sigma \text{ is $(1-c)$-feasible for } I} \leq \exp\left(-\frac{9}{128} n^2 C^4 \right).
    \end{equation}
    
    \medskip 
    \noindent \textbf{Concluding:}
    We finish our argument with a union bound over the sets $I \subseteq [n]$ for which $|I| \geq Cn$. There are at most $2^n < e^n$ such sets; therefore by \eqref{eq:feasible-assn-bound},
    \begin{align*}
        \probover{v \sim \D}{ \bigvee_{I \subseteq [n]: |I| \geq C n} \left( \exists \sigma: \sigma \text{ is $(1-c)$-feasible for } I \right) }
        &\leq \sum_{I \subseteq [n]: |I| \geq C n}  \probover{v\sim\D}{\exists \sigma: \sigma \text{ is $(1-c)$-feasible for } I }  \\
        &\leq \exp\left(n -\frac{9}{128} \cdot n^2 C^4 \right) \\
        &=O(e^{-n^2/\log^2 n})
    \end{align*}
    for our choice of $C = O(\log^{-1/2}(n))$. The left-hand side is precisely the probability of the event $\mathcal{E}(c,C,n)$, and so this demonstrates the stated claim.
\end{proof}

With this lemma in hand, \Cref{lem:max-t-feasible-num-buyers} follows as a direct corollary. 

\xosfourthlemma*
\begin{proof}
    Since $t = \sqrt{\log n}$ for our choice of parameters, observe that any $I$ which is $(1-\frac{1}{t+1})$-feasible under some assignment $\sigma$ is also $(1-\frac{2}{t})$-feasible under $\sigma$. 
    
    We may take therefore take $\gamma = 2$ and $c = \frac{2}{\log^{1/2} n}$ and $C = \frac{16}{\log^{1/2} n}$ and apply \Cref{lem:approx-packing}, which implies that
    \[
        \probover{v \sim \D}{\mathcal{E}(c,C,n)} = \bigo{e^{-n^2/\log^{2} n}}.
    \]
    Therefore, turning to $|I_{max}(v)|$, we have
    \begin{align*}
        \xpectover{v}{|I_{max}(v)|} &= \xpectover{v}{|I_{max}(v)| \mid \neg \mathcal{E}(c,C,n)} \probover{v}{\neg \mathcal{E}(c,C,n)} + \xpectover{v}{|I_{max}(v)| \mid \mathcal{E}(c,C,n)} \probover{v}{ \mathcal{E}(c,C,n)} \\
        &\leq Cn + n \cdot \probover{v}{ \mathcal{E}(c,C,n)} \\
        &\leq n \cdot \frac{16}{t} + n \cdot c_1 \cdot e^{-n^2/\log^{2} n} \leq 17 \cdot \frac{n}{t}
    \end{align*}
    for some constant $c_1$ and for sufficiently large $k=n$, as claimed.
\end{proof}

%% file: main.bbl
\begin{thebibliography}{26}
\providecommand{\natexlab}[1]{#1}
\providecommand{\url}[1]{\texttt{#1}}
\expandafter\ifx\csname urlstyle\endcsname\relax
  \providecommand{\doi}[1]{doi: #1}\else
  \providecommand{\doi}{doi: \begingroup \urlstyle{rm}\Url}\fi

\bibitem[Alaei(2014)]{doi:10.1137/120878422}
Saeed Alaei.
\newblock Bayesian combinatorial auctions: Expanding single buyer mechanisms to
  many buyers.
\newblock \emph{SIAM Journal on Computing}, 43\penalty0 (2):\penalty0 930--972,
  2014.

\bibitem[Babaioff et~al.(2014)Babaioff, Immorlica, Lucier, and
  Weinberg]{babaioff2014simple}
Moshe Babaioff, Nicole Immorlica, Brendan Lucier, and S~Matthew Weinberg.
\newblock A simple and approximately optimal mechanism for an additive buyer.
\newblock In \emph{2014 IEEE 55th Annual Symposium on Foundations of Computer
  Science}, pages 21--30. IEEE, 2014.

\bibitem[Briest et~al.(2015)Briest, Chawla, Kleinberg, and
  Weinberg]{BCKW-JET15}
Patrick Briest, Shuchi Chawla, Robert Kleinberg, and S.~Matthew Weinberg.
\newblock Pricing lotteries.
\newblock \emph{Journal of Economic Theory}, 156:\penalty0 144--174, 2015.

\bibitem[Cai and Zhao(2017)]{10.1145/3055399.3055465}
Yang Cai and Mingfei Zhao.
\newblock Simple mechanisms for subadditive buyers via duality.
\newblock In \emph{Proceedings of the 49th Annual ACM SIGACT Symposium on
  Theory of Computing}, STOC 2017, page 170–183. Association for Computing
  Machinery, 2017.

\bibitem[Cai et~al.(2019)Cai, Devanur, and Weinberg]{cai2019duality}
Yang Cai, Nikhil~R Devanur, and S~Matthew Weinberg.
\newblock A duality-based unified approach to bayesian mechanism design.
\newblock \emph{SIAM Journal on Computing}, 50\penalty0 (3):\penalty0
  STOC16--160, 2019.

\bibitem[Chalermsook et~al.(2013)Chalermsook, Laekhanukit, and
  Nanongkai]{chalermsook2013independent}
Parinya Chalermsook, Bundit Laekhanukit, and Danupon Nanongkai.
\newblock Independent set, induced matching, and pricing: Connections and tight
  (subexponential time) approximation hardnesses.
\newblock In \emph{2013 IEEE 54th Annual Symposium on Foundations of Computer
  Science}, pages 370--379. IEEE, 2013.

\bibitem[Chawla and Miller(2016)]{chawla2016mechanism}
Shuchi Chawla and J~Benjamin Miller.
\newblock Mechanism design for subadditive agents via an ex ante relaxation.
\newblock In \emph{Proceedings of the 2016 ACM Conference on Economics and
  Computation}, pages 579--596, 2016.

\bibitem[Chawla et~al.(2010)Chawla, Hartline, Malec, and
  Sivan]{chawla2010multi}
Shuchi Chawla, Jason~D Hartline, David~L Malec, and Balasubramanian Sivan.
\newblock Multi-parameter mechanism design and sequential posted pricing.
\newblock In \emph{Proceedings of the forty-second ACM symposium on Theory of
  computing}, pages 311--320, 2010.

\bibitem[Chawla et~al.(2015)Chawla, Malec, and Sivan]{chawla2015power}
Shuchi Chawla, David Malec, and Balasubramanian Sivan.
\newblock The power of randomness in bayesian optimal mechanism design.
\newblock \emph{Games and Economic Behavior}, 91:\penalty0 297--317, 2015.

\bibitem[Chawla et~al.(2022{\natexlab{a}})Chawla, Rezvan, Teng, and
  Tzamos]{10.1145/3519935.3520065}
Shuchi Chawla, Rojin Rezvan, Yifeng Teng, and Christos Tzamos.
\newblock Pricing ordered items.
\newblock In \emph{Proceedings of the 54th Annual ACM SIGACT Symposium on
  Theory of Computing}, STOC 2022, page 722–735. Association for Computing
  Machinery, 2022{\natexlab{a}}.

\bibitem[Chawla et~al.(2022{\natexlab{b}})Chawla, Teng, and
  Tzamos]{chawla2022buy}
Shuchi Chawla, Yifeng Teng, and Christos Tzamos.
\newblock Buy-many mechanisms are not much better than item pricing.
\newblock \emph{Games Econ. Behav.}, 134:\penalty0 104--116,
  2022{\natexlab{b}}.

\bibitem[Chawla et~al.(2023)Chawla, Rezvan, Teng, and Tzamos]{chawla23buy}
Shuchi Chawla, Rojin Rezvan, Yifeng Teng, and Christos Tzamos.
\newblock Buy-many mechanisms for many unit-demand buyers.
\newblock In \emph{Web and Internet Economics - 19th International Conference,
  {WINE} 2023}, volume 14413, pages 21--38. Springer, 2023.

\bibitem[Correa and Cristi(2023)]{10.1145/3564246.3585151}
Jos\'{e} Correa and Andr\'{e}s Cristi.
\newblock A constant factor prophet inequality for online combinatorial
  auctions.
\newblock In \emph{Proceedings of the 55th Annual ACM Symposium on Theory of
  Computing}, STOC 2023, page 686–697. Association for Computing Machinery,
  2023.

\bibitem[D\"{u}tting et~al.(2020{\natexlab{a}})D\"{u}tting, Feldman,
  Kesselheim, and Lucier]{doi:10.1137/20M1323850}
Paul D\"{u}tting, Michal Feldman, Thomas Kesselheim, and Brendan Lucier.
\newblock Prophet inequalities made easy: Stochastic optimization by pricing
  nonstochastic inputs.
\newblock \emph{SIAM Journal on Computing}, 49\penalty0 (3):\penalty0 540--582,
  2020{\natexlab{a}}.

\bibitem[D\"{u}tting et~al.(2020{\natexlab{b}})D\"{u}tting, Kesselheim, and
  Lucier]{doi:10.1137/20M1382799}
Paul D\"{u}tting, Thomas Kesselheim, and Brendan Lucier.
\newblock An $o(\log \log m)$ prophet inequality for subadditive combinatorial
  auctions.
\newblock \emph{SIAM Journal on Computing}, pages 239--275, 2020{\natexlab{b}}.

\bibitem[Feldman et~al.(2015)Feldman, Gravin, and
  Lucier]{doi:10.1137/1.9781611973730.10}
Michal Feldman, Nick Gravin, and Brendan Lucier.
\newblock Combinatorial auctions via posted prices.
\newblock In \emph{Proceedings of the Twenty-Sixth Annual ACM-SIAM Symposium on
  Discrete Algorithms}, SODA '15, page 123–135. Society for Industrial and
  Applied Mathematics, 2015.

\bibitem[Feldman et~al.(2016)Feldman, Svensson, and Zenklusen]{feldman16online}
Moran Feldman, Ola Svensson, and Rico Zenklusen.
\newblock Online contention resolution schemes.
\newblock In \emph{Proceedings of the Twenty-Seventh Annual {ACM-SIAM}
  Symposium on Discrete Algorithms, {SODA} 2016}, pages 1014--1033. {SIAM},
  2016.

\bibitem[Hajiaghayi et~al.(2007)Hajiaghayi, Kleinberg, and
  Sandholm]{10.5555/1619645.1619656}
Mohammad~Taghi Hajiaghayi, Robert Kleinberg, and Tuomas Sandholm.
\newblock Automated online mechanism design and prophet inequalities.
\newblock In \emph{Proceedings of the 22nd National Conference on Artificial
  Intelligence - Volume 1}, AAAI'07, page 58–65. AAAI Press, 2007.

\bibitem[Hart and Nisan(2019)]{hart2013menu}
Sergiu Hart and Noam Nisan.
\newblock Selling multiple correlated goods: Revenue maximization and menu-size
  complexity.
\newblock \emph{Journal of Economic Theory}, 183:\penalty0 991--1029, 2019.

\bibitem[Joag-Dev and Proschan(1983)]{joag1983negative}
Kumar Joag-Dev and Frank Proschan.
\newblock Negative association of random variables with applications.
\newblock \emph{The Annals of Statistics}, pages 286--295, 1983.

\bibitem[Kleinberg and Weinberg(2012)]{10.1145/2213977.2213991}
Robert Kleinberg and Seth~Matthew Weinberg.
\newblock Matroid prophet inequalities.
\newblock In \emph{Proceedings of the Forty-Fourth Annual ACM Symposium on
  Theory of Computing}, STOC '12, page 123–136. Association for Computing
  Machinery, 2012.

\bibitem[Lee and Singla(2018)]{lee18optimal}
Euiwoong Lee and Sahil Singla.
\newblock Optimal online contention resolution schemes via ex-ante prophet
  inequalities.
\newblock In \emph{26th Annual European Symposium on Algorithms, {ESA} 2018},
  volume 112 of \emph{LIPIcs}, pages 57:1--57:14. Schloss Dagstuhl -
  Leibniz-Zentrum f{\"{u}}r Informatik, 2018.

\bibitem[Lucier(2017)]{10.1145/3144722.3144725}
Brendan Lucier.
\newblock An economic view of prophet inequalities.
\newblock \emph{SIGecom Exch.}, 16\penalty0 (1):\penalty0 24–47, sep 2017.

\bibitem[Rubinstein and Weinberg(2018)]{rubinstein2018simple}
Aviad Rubinstein and S~Matthew Weinberg.
\newblock Simple mechanisms for a subadditive buyer and applications to revenue
  monotonicity.
\newblock \emph{ACM Transactions on Economics and Computation (TEAC)},
  6\penalty0 (3-4):\penalty0 1--25, 2018.

\bibitem[Vondr{\'{a}}k et~al.(2011)Vondr{\'{a}}k, Chekuri, and
  Zenklusen]{vondrak11submodular}
Jan Vondr{\'{a}}k, Chandra Chekuri, and Rico Zenklusen.
\newblock Submodular function maximization via the multilinear relaxation and
  contention resolution schemes.
\newblock In \emph{Proceedings of the 43rd {ACM} Symposium on Theory of
  Computing, {STOC} 2011}, pages 783--792. {ACM}, 2011.

\bibitem[Zhang(2022)]{ZHANG2022143}
Hanrui Zhang.
\newblock Improved prophet inequalities for combinatorial welfare maximization
  with (approximately) subadditive agents.
\newblock \emph{Journal of Computer and System Sciences}, 123:\penalty0
  143--156, 2022.

\end{thebibliography}
